\documentclass[12pt]{article}
\usepackage[margin=1in]{geometry}
\usepackage{amsmath,amssymb,amsthm,amsfonts}
\usepackage{graphicx}
\usepackage{booktabs}
\usepackage{natbib}
\usepackage[ruled,vlined]{algorithm2e}
\usepackage{enumitem}
\usepackage[hyperfootnotes=false]{hyperref}
\hypersetup{colorlinks=true,citecolor=blue,linkcolor=black,urlcolor=black,filecolor=black}
\usepackage{setspace}
\usepackage{microtype}
\theoremstyle{plain}
\newtheorem{theorem}{Theorem}
\newtheorem{proposition}{Proposition}
\newtheorem{lemma}{Lemma}
\newtheorem{corollary}{Corollary}
\theoremstyle{definition}

\newcommand{\E}{\mathbb{E}}

\newcommand{\Var}{\operatorname{Var}}
\newcommand{\Cov}{\operatorname{Cov}}
\newcommand{\Ni}{\mathcal{N}_i}

\title{Whom Should a Platform Amplify? Truth, Engagement, and Networked Polarization}
\author{Zaruhi Hakobyan\thanks{email: zaruhi.hakobyan@uni.lu. I thank Vladimir Asiryan, Matthew O. Jackson, Konstantin Sonin, Alexei Zakharov, Sergei Guriev,  Kemal Kivanc Akoz, Eren Arbatli, Benjamin Holcblat, Christos Koulovatianos, and Roberto Steri for their valuable feedback. I am also grateful to the participants of the job market course at the University of Luxembourg, the 2024 Dynamic Games and Applications workshop, LAGV 2024, the Ninth Annual Conference on Network Science and Economics (Minneapolis, 2024), the Armenian Economic Association, the 2023 Economics and Management of Networks conference (Palermo), and the ARS'23 Ninth International Workshop on Social Network Analysis. I used ChatGPT, Claude and Refine.ink for proofreading. The usual disclaimer applies. I acknowledge financial support from the FNR CORE C23/SC/18098569/POPULISM grant project. This manuscript substantially revises an earlier working paper ``Can a social planner manipulate network dynamics and solve coordination problems?'' (SSRN No.~3587024).}\\[2pt]
\normalsize University of Luxembourg}
\begin{document}
\pagestyle{plain}
\setcounter{page}{1}
\maketitle

\begin{abstract}
Social-media platforms allocate attention by deciding whose content receives
reach. We study this problem in a network coordination model in which users
learn from observed signals but also care about coordinating with others.
Amplifying a source, therefore, spreads information while potentially making that
source a more influential coordination reference. For any finite directed
network, we show that the effect of an intervention on truth-tracking accuracy
decomposes into an information gain and a change in network-amplified bias
risk. The latter depends on how sources' influence profiles interact with the
network's entire bias composition, so the most informative source need not be
the best source to amplify. Broadcasting can improve the use of common
information, while amplification of directionally biased sources can reduce
accuracy. Integrated communities can outperform segregated ones; users may
reject informative cross-cutting links, and engagement-oriented platforms may
favor connections that reduce collective accuracy.

\end{abstract}

\newpage
\section{Introduction}\label{sec:intro}

Social-media platforms allocate attention by deciding whose content appears in users' feeds and which sources receive greater reach. Recommendation systems determine which posts become widely visible and which users remain comparatively obscure. These choices are often guided by engagement---time, clicks, likes, and other measures of user activity---rather than by the accuracy of the beliefs or actions that emerge on the platform \citep{zhuravskaya2020,aridor2024}. They can also materially change exposure. Evidence from Twitter shows that ranking algorithms amplify some political content more than others \citep{huszar2022}, while experimental evidence from Facebook shows that feed algorithms affect both the composition of users' information environments and their activity on the platform \citep{guess2023}. Platforms can also intervene directly in the reach of particular users. In 2023, Twitter engineers reportedly modified the recommendation algorithm to increase the visibility of Elon Musk's posts, including to users who did not follow him \citep{paul2023}. These examples illustrate a basic design margin: a platform can change who sees a source without changing what that source knows. This paper asks when reallocating reach improves collective accuracy, whom an accuracy-oriented platform should amplify, and how these choices change when the platform instead values engagement.

The central trade-off is that amplification can affect behavior through two channels. First, showing a source to more users spreads the information contained in that source's signal. Second, when users care about coordinating with others, greater exposure can also make the source a more important social reference. Amplifying a precise source may therefore improve learning, while amplifying a directionally biased source may simultaneously pull actions away from the state. The same tension arises more generally whenever a platform changes who observes whom in an environment with dispersed information and strategic interaction: changing exposure can alter both what agents know and whose actions become focal for coordination.

We study this problem in a network coordination game with dispersed information and heterogeneous directional preferences. Users want their actions to track an unknown fundamental---the correct assessment of a news event, a policy, or a market condition---but they also care about coordinating with others, as in the beauty-contest motive of \citet{keynes1936} and \citet{morrisshin2002}. Each user observes a public signal, her own private signal, and the private signals of sources appearing in her feed. A directed link $i\to j$ means that user $i$ observes source $j$. In the baseline, the same link also makes $j$ part of the reference group relevant for $i$'s coordination payoff. Users may prefer actions shifted in different directions, but these preferences do not distort their beliefs: all observed signals are processed according to Bayes' rule. The platform changes only the observation network, and therefore the allocation of reach across sources.

The equilibrium separates the two forces analytically. For every fixed network, the unique linear equilibrium consists of signal loadings that do not depend on users' directional preferences and a directional component that does not depend on signal precision. This separation allows us to ask what a network intervention does to information and what it does to directional influence without mixing the two effects. Our main result gives an exact characterization for any finite directed network. Moving from a baseline network $M_0$ to a new network $M_1$ changes truth-tracking accuracy according to
\textit{information gain
minus
change in quadratic bias risk}.

The first term captures the reduction in uncertainty generated by the new pattern of exposure. The second captures how the intervention changes the propagation of users' directional preferences through the coordination network.

The bias term has a simple network interpretation. Each source has an equilibrium influence profile describing how a unit directional preference at that source affects the actions of users throughout the network. Collecting the overlaps among these profiles produces an \emph{influence-overlap matrix}. Its diagonal elements measure how strongly each source's own directional preference propagates, while its off-diagonal elements measure whether the influence of different sources reinforces or offsets one another. Same-sign preferences reinforce through these cross terms; opposite-sign preferences can partially cancel. The relevant object is therefore not the centrality of an amplified source in isolation, but how that source's influence fits into the bias composition already present in the network.

This yields a central source-selection implication. The best source to amplify cannot generally be identified from signal precision or own bias alone. A highly informative source may be unattractive when amplification makes her directional preference reinforce distortions already present in the network. Conversely, a source with a larger own bias can be less harmful when her influence offsets existing distortions. When only one source is biased, these interactions disappear and the criterion reduces to a transparent threshold: amplification is desirable when its information gain is large enough to compensate for the increase in the source's network-amplified bias cost. The corresponding influence profile has a discounted-walk representation closely related to Katz--Bonacich influence \citep{katz1953,bonacich1987,ballester2006,jackson2008}.

We next use several benchmarks to show how this general trade-off operates. First, we isolate the information channel in a symmetric global-broadcast environment with uniform coordination weights. The benchmark connects to the transparency result of \citet{morrisshin2002}, in which more precise public information can reduce truth-tracking accuracy because agents place excessive strategic weight on a common signal. Broadcasting one private signal creates a second signal that is common in use. In the symmetric benchmark, this is equivalent for accuracy to increasing effective common precision by the precision of the broadcast signal, and it eliminates the Morris--Shin non-monotone region. The result contributes to the literature on the social value of public information \citep{hellwig2005,svensson2006,morrisshin2006,angeletospavan2007,cornandheinemann2008,jameslawler2011} by showing how the allocation of reach can change transparency comparative statics without changing the precision of the original public signal.

Second, we introduce scarce broadcast capacity. With no directional preferences and no attention constraint, the complete network is the accuracy first best: the platform simply makes every signal available to everyone. When only a limited number of sources can receive broadcast reach, source selection becomes meaningful. In the symmetric large-$N$ benchmark, accuracy depends on the amplified set through total broadcast precision, so in the monotone region the platform allocates scarce reach to the most informative sources. This benchmark is related to key-player and targeting problems in networks \citep{ballester2006,sakovicssteiner2012,galeotti2020} and to information use by networked agents \citep{golubmorris2017,myattwallace2019,leister2019}. More recent work studies strategic targeting when influence depends on network position and agents' initial attitudes. \citet{vohra2026} studies how strategic influencers choose whom to target when agents differ in network position and bias, while \citet{comola2026} examine competing influence through strategic targeting in networks. Our intervention is different: the platform allocates \emph{reach} by changing whose information becomes visible and, in the coupled model, who becomes socially salient.

Third, reintroducing directional preferences shows why dense exposure need not imply accurate aggregation. In a two-community benchmark with equal and opposite biases, two internally complete but segregated communities perform worse than an integrated complete network even though every user observes every signal produced within her own community. Integration improves accuracy twice: it expands the information set and weakens the directional effect of coordinating only with same-sign users. This distinction is closely related to the broader literature on selective exposure and algorithmic curation. Empirical work shows that social-media algorithms can shape exposure to ideologically congruent or counter-attitudinal content \citep{bakshy2015,levy2021,guess2023,nyhan2023}, while theoretical work studies how curation affects connectivity, filter bubbles, and polarization \citep{bermankatona2020}. Our mechanism is different: users remain Bayesian and signals are exogenous; polarization changes because amplification can alter the social-reference structure through which directional preferences affect actions.

Finally, we relax two assumptions behind direct network design. First, an accuracy-oriented platform may know which links would improve collective truth tracking without being able to impose them. We therefore let the platform recommend feed changes while users decide whether to accept them. Without directional preferences, users accept nearly all accuracy-improving recommendations in the baseline simulations. With directional preferences, however, they frequently reject informative cross-cutting links because the same links that improve collective accuracy can make the receiving user's coordination environment less attractive. Second, the platform itself may value engagement rather than truth tracking. We use average expected user utility as a reduced-form engagement objective and show that engagement and accuracy can rank the same amplification differently. An engagement-oriented platform may favor validating, same-sign connections even when they reduce collective accuracy. This connects our analysis to a growing literature that treats recommendation and matching rules as economic design instruments. \citet{qianjain2024} study how recommendation systems affect content creation, \citet{peitzsobolev2026} show that personalized product recommendations can be welfare reducing, \citet{xuyang2026} study recommendations jointly with dynamic pricing and consumer feedback, and \citet{condorelliszentes2026} study how platform matching rules shape surplus. Relatedly, \citet{filippas2025} jointly model content production and consumption on social media, while \citet{sideriusdebo2026} study an upstream regulatory intervention that constrains the computational resources available to recommendation systems in a model of social-media addiction. Our instrument is different: the platform reallocates reach across information sources, and a feed change may also change the social-reference structure through which users coordinate. Separating information exposure from coordination salience then clarifies the mechanism: in the model, purely informational rewiring can improve learning, but the polarization channel operates through changes in social-reference relationships.

The paper contributes to three literatures. First, it contributes to network games and targeting by treating the observation network itself as the intervention and deriving source-specific statistics for the value of reach \citep{ballester2006,jackson2008,galeotti2020,vohra2026,comola2026}. Second, it connects network design to information aggregation and beauty-contest coordination \citep{morrisshin2002,golubmorris2017,myattwallace2019}. Third, it contributes to the economics of social media and information intermediaries. Existing work studies misinformation diffusion and strategic information transmission \citep{allcottgentzkow2017,vosoughi2018,mostagir2022,guriev2023,acemoglu2024,crawfordsobel1982,kamenicagentzkow2011,mullainathanshleifer2005,edmond2013}, algorithmic curation and recommendation design \citep{bermankatona2020,qianjain2024,peitzsobolev2026,xuyang2026,condorelliszentes2026,sideriusdebo2026}, and endogenous content production \citep{filippas2025,kanikhakobyan2026}. Here, signals and their precision are exogenous, and users do not strategically misreport them. The platform's instrument is the allocation of reach: it changes who sees whom. This also distinguishes the mechanism from social-image models of public expression such as \citet{braghieri2024}. The closest platform mechanism is \citet{acemoglu2024}. Their analysis shows how engagement incentives can distort information transmission by encouraging misinformation sharing. Our paper is complementary: even when signals are truthful and users update according to Bayes' rule, amplification can reduce accuracy because it changes who becomes influential for coordination. Together, the two mechanisms show that engagement-oriented platform design can generate distortions both through the information that is transmitted and through the network influence created by exposure.

The rest of the paper proceeds as follows. Section~\ref{sec:model} introduces the model and establishes the equilibrium separation result. Section~\ref{sec:svi} studies accuracy-oriented network design, beginning with broadcast benchmarks and then deriving the general finite-network information--bias decomposition, source-ranking criterion, and demagogue threshold. Section~\ref{sec:implementation} relaxes direct platform control and studies implementation under user choice. Section~\ref{sec:engagement} replaces the accuracy objective with engagement and separates information exposure from coordination salience. Section~\ref{sec:discussion} concludes.

\section{Model}\label{sec:model}

There are $N<\infty$ users on a social-media platform, indexed by
$i=1,\ldots,N$. The platform determines whose content appears in whose feed.
A feed link can matter for two reasons. First, it gives the receiving user an
additional source of information. Second, in the baseline model, appearing in a
user's feed can also make the source more important for that user's coordination
decision. We refer to these as the \emph{information channel} and the
\emph{social-reference channel}, respectively.

The feed network is directed and represented by an adjacency matrix
$M=(M_{ij})_{i,j=1}^N$, where $M_{ij}\in\{0,1\}$ and $M_{ii}=0$.
The entry $M_{ij}=1$ means that user $i$ receives the private signal, or content,
of user $j$. Thus, links point from the user who receives information to the
source who provides it. Let
$\Ni(M)=\{j\neq i:M_{ij}=1\}$
denote the set of sources appearing in user $i$'s feed.

This direction of links gives a simple interpretation of amplification. For any
source $j$, $\sum_{i=1}^N M_{ij}$ is the number of users whose feeds contain
$j$'s signal. We refer to this number as source $j$'s \emph{reach}. Adding a
link $i\to j$ increases the reach of source $j$ by one user. At the extreme, if
every other user observes $j$, then $j$'s signal is broadcast to the entire
population. The platform's design problem can therefore be viewed as a problem
of allocating reach across sources.

Users are uncertain about an underlying state $\theta\in\mathbb R$. Depending
on the application, $\theta$ may represent the correct assessment of a news
event, the quality of a policy, or the value of an asset.\footnote{Users share a diffuse
prior over $\theta$, which we interpret as the limit of a Gaussian prior whose
precision converges to zero. This assumption allows the analysis to focus on
the information contained in observed signals rather than on a common
informative prior. Appendix~\ref{app:exist} provides the formal limiting
argument.}

Everyone observes a public signal
\begin{equation}\label{eq:public-signal}
y=\theta+\eta,\qquad \eta\sim\mathcal N(0,\alpha^{-1}).
\end{equation}
The parameter $\alpha>0$ is the precision of public information: a larger $\alpha$ means that the public signal is more informative about $\theta$.

Each user $i$ also has a private signal
\begin{equation}\label{eq:private-signal}
x_i=\theta+\varepsilon_i,\qquad \varepsilon_i\sim\mathcal N(0,\beta_i^{-1}),
\end{equation}
where $\beta_i>0$ is user $i$'s private-signal precision. A larger $\beta_i$ therefore means that user $i$ is a more informative source. All signal noises are mutually independent.

The network determines which private signals a user can additionally observe. User $i$ sees the public signal $y$, her own private signal $x_i$, and the private signals $\{x_j:j\in\Ni(M)\}$ of the sources in her feed. We use $\E_i[\cdot]$ as shorthand for expectations conditional on exactly these observations. Under the diffuse prior,
\begin{equation}\label{eq:posterior}
\E_i[\theta]
=
\frac{\alpha y+\beta_i x_i+\sum_{j\in\Ni(M)}\beta_jx_j}
{\alpha+\beta_i+\sum_{j\in\Ni(M)}\beta_j}.
\end{equation}
Equation~\eqref{eq:posterior} shows the direct informational value of reach. Adding a source $j$ to user $i$'s feed gives $i$ one additional signal, and that signal is more valuable when $\beta_j$ is larger. The platform does not choose signal realizations or signal precision; it changes only who observes whom.

\paragraph{Preferences and coordination}\label{sec:model-preferences}

Each user chooses an action $a_i\in\mathbb R$. The action can be interpreted as a forecast, judgment, position, or other decision that should track the underlying state. Users care about two things: choosing an action close to their own preferred assessment and coordinating with other users.

User $i$ has a signed directional preference $b_i\in\mathbb R$. Her ideal action, if the state were known and coordination did not matter, would be $\theta+b_i$. Positive and negative values of $b_i$ represent preferences in opposite directions, while $b_i=0$ describes an unbiased user. Importantly, $b_i$ affects preferences rather than beliefs: users process all observed signals according to Bayes' rule. In the two-camp exercises below, we specialize to
\[
b_i\in\{-b,+b\},\qquad b>0,
\]
so that the two groups have equally strong but oppositely directed preferences.

Given an action profile $a=(a_1,\ldots,a_N)$, user $i$'s payoff is
\begin{equation}\label{eq:utility}
u_i(a,\theta)
=
-(1-r_i)\big(a_i-(\theta+b_i)\big)^2
-r_i\sum_{j\neq i}\phi_{ij}(M)(a_j-a_i)^2.
\end{equation}
The first term captures the user's desire to choose an action close to her own ideal $\theta+b_i$. The second is a beauty-contest motive: the user also dislikes choosing an action that differs from the actions of socially relevant others \citep{keynes1936,morrisshin2002}. The parameter $r_i\in(0,1)$ determines the strength of this coordination motive. When $r_i$ is small, the user's own assessment is relatively important. As $r_i$ rises, matching the actions of others becomes more important.

The weights $\phi_{ij}(M)$ determine \emph{whose} actions user $i$ wants to coordinate with. We allow users appearing in the feed to be more or less socially salient than users outside the feed. The parameter $q_i\in[0,1]$ is the total coordination weight that user $i$ places on sources in her feed. The remaining weight, $1-q_i$, is placed on users outside the feed. Within each group, weight is distributed uniformly:
\begin{equation}\label{eq:phi}
\phi_{ij}(M)=
\begin{cases}
q_i/|\Ni(M)|, & j\in\Ni(M),\\[2mm]
(1-q_i)/(N-1-|\Ni(M)|), & j\notin\Ni(M),\ j\neq i.
\end{cases}
\end{equation}
Thus $q_i$ measures the \emph{social salience of the feed}. A high $q_i$ means that users appearing in $i$'s feed receive disproportionate weight in her coordination decision. If both groups are nonempty, the value
$q_i=\frac{|\Ni(M)|}{N-1}
$
makes all other users equally important; values above this benchmark overweight feed sources, while values below it overweight users outside the feed. If one of the two groups is empty, we drop the empty component and assign equal weight $1/(N-1)$ to the nonempty group.

The roles of $\beta_j$ and $q_i$ are therefore distinct. The parameter $\beta_j$ measures how informative source $j$ is. The parameter $q_i$ measures how much social weight user $i$ places on sources appearing in her feed. Neither parameter changes the other: $q_i$ does not distort Bayesian learning, and $\beta_j$ does not directly determine social salience.

This distinction is central to the paper. In the baseline, the same feed network determines both information exposure and social reference groups. Adding a link $i\to j$ therefore does two things at once: it gives user $i$ access to $j$'s signal and can make $j$ more important for $i$'s coordination decision. A precise source can improve learning through the first channel, while a directionally biased source can move actions away from the state through the second. Section~\ref{sec:svi} first isolates the information channel using uniform coordination weights. The later polarization results use the coupled benchmark, and Section~\ref{sec:decouple} separates the two channels explicitly. The primitives $(r_i,q_i,b_i,\beta_i)$ are common knowledge.

\paragraph{Timing and equilibrium}\label{sec:equilibrium}

Fix a network $M$, which is observed by all users. The stage game proceeds as follows:
\begin{enumerate}[label=(\roman*)]
\item the state $\theta$, the public signal $y$, and the private signals $(x_1,\ldots,x_N)$ are realized;
\item each user observes the signals available through her feed;
\item users simultaneously choose actions $a_i\in\mathbb R$;
\item payoffs are realized according to \eqref{eq:utility}.
\end{enumerate}

The game is static. The index $t$ introduced later in Section~\ref{sec:implementation} labels rounds of the platform's recommendation procedure; it is not economic time in a repeated game. Each round simply evaluates a new network using a fresh instance of the static game above.

Given her information, user $i$ chooses $a_i$ to maximize expected utility. The first-order condition gives
\begin{equation}\label{eq:bestreply}
a_i
=
(1-r_i)\big(\E_i[\theta]+b_i\big)
+r_i\sum_{j\neq i}\phi_{ij}(M)\E_i[a_j].
\end{equation}
The first term is user $i$'s posterior ideal action. The second is her forecast of the actions of the users with whom she wants to coordinate. Equation~\eqref{eq:bestreply} therefore summarizes the economic tension in the model: users want to follow information, but they also want to coordinate.

The equation also makes clear why the platform can matter. Changing the network can change $\E_i[\theta]$ by changing the signals user $i$ observes. In the coupled model it can also change the coordination weights $\phi_{ij}(M)$, and hence the higher-order expectations $\E_i[a_j]$. If $r_i=0$, only the information channel remains and user $i$ chooses $\E_i[\theta]+b_i$. If in addition $b_i=0$, she simply chooses the Bayesian posterior mean.

The diffuse-prior formulation and the formal equilibrium strategy space are treated in Appendix~\ref{app:exist}.

The equilibrium is linear. We can write user $i$'s action as
\begin{equation}\label{eq:linear-strategy}
a_i
=
\sum_{k=1}^N c_{ik}x_k+c_{iy}y+d_i,
\end{equation}
where $c_{ik}=0$ whenever user $i$ does not observe $x_k$. The coefficients $c_{ik}$ and $c_{iy}$ describe how equilibrium actions respond to information. The constant $d_i$ captures the directional component of user $i$'s action generated by the bias parameters.

\begin{theorem}\label{thm:exist}
Fix any directed network $M$. If
$
\bar r=\max_i r_i<1,
$
 the game has a unique linear Bayes–Nash equilibrium.  The signal loadings $\{c_{ik},c_{iy}\}$ are independent of the bias parameters $\{b_i\}$. The directional components $\{d_i\}$ are independent of the signal precisions $(\alpha,\{\beta_i\})$.\footnote{The proof and the formal strategy space are given in Appendix~\ref{app:exist}.}
\end{theorem}

Theorem~\ref{thm:exist} separates the two channels analytically. Information quality and network exposure determine the signal loadings. Directional preferences and the coordination structure determine the constants $d_i$. This allows us to study the informational benefit of amplification separately from the bias that amplification may transmit through social coordination.

The directional component has a particularly intuitive fixed-point representation. Matching the constant terms in \eqref{eq:bestreply} gives
\begin{equation}\label{eq:bias-scalar}
d_i=(1-r_i)b_i+r_i\sum_{j\neq i}\phi_{ij}(M)d_j.
\end{equation}
Thus $d_i$ is not generally equal to user $i$'s own bias $b_i$. It combines her own directional preference with the directional components of the users who matter for coordination. Bias can therefore travel indirectly through chains of social-reference relationships.

Let $R=\operatorname{diag}(r_i)$, let $\Phi(M)$ collect the weights $\phi_{ij}(M)$, and let $b=(b_1,\ldots,b_N)'$ and $d=(d_1,\ldots,d_N)'$. Equation~\eqref{eq:bias-scalar} can then be written compactly as
\begin{equation}\label{eq:bias-solution}
d=(I-R\Phi(M))^{-1}(I-R)b.
\end{equation}

Equation~\eqref{eq:bias-scalar} captures the main mechanism of the bias channel: coordination transmits users' directional preferences through the network of social references. Thus, a user's equilibrium directional component reflects not only her own bias but also the biases of users who matter for her coordination decision, including indirect effects transmitted through chains of reference relationships. Section~\ref{sec:bias} develops this network interpretation using Katz--Bonacich influence profiles.

\section{Accuracy-Oriented Platform}\label{sec:svi}

This section characterizes the accuracy consequences of allocating reach. We begin with tractable broadcast benchmarks that isolate the information channel and show how common reach changes transparency and source selection. We then reintroduce directional preferences and move to arbitrary finite networks, where amplification can also change coordination influence. The resulting information--bias decomposition provides the general criterion for evaluating network interventions. For any fixed network $M$, Theorem~\ref{thm:exist} determines users' equilibrium actions, and the platform compares networks according to how closely those actions track the true state $\theta$.

\paragraph{Platform objective}\label{sec:platform-benchmark}

The platform first chooses the network $M$. Users then observe their signals and choose actions according to the game described in Section~\ref{sec:model}. The platform evaluates the resulting network using
\begin{equation}\label{eq:accuracy}
W(M)
=
-\frac{1}{N}\sum_{i=1}^N
\mathbb E\big[(a_i-\theta)^2\big].
\end{equation}
A larger $W(M)$ means that, on average, users' actions are closer to the true state.  $W(M)$ is a \emph{truth-tracking accuracy index} and not a welfare measure. The platform does not value coordination itself or users' gains from moving toward their preferred actions. It cares only about how accurately equilibrium actions track $\theta$. Section~\ref{sec:engagement} later considers a platform that instead values users' expected utility. The platform can only change the network structure. In particular, it can change who observes whom (feed composition) and, therefore, how much reach each source receives. It cannot create new signals, change their realizations or precision, or choose users' actions directly.

In this section, we begin with benchmark environments in which the platform
directly chooses who observes whom. We first study the effect of broadcasting a
private signal to the population when $b_i=0$, and then ask which user's signal should be broadcast. We next reintroduce directional preferences $b_i \neq 0$ and return to
arbitrary finite networks. In this more general setting, the effect of a network
intervention on accuracy separates into an information gain and a change in bias
risk. Section~\ref{sec:implementation} then relaxes direct platform control in feed composition, and the 
platform recommends link changes, but users decide whether to accept or reject
them.

\paragraph{Signal amplification and the transparency reversal}\label{sec:reversal}

We first isolate the information effect by setting $b_i=0$
and using uniform coordination weights. Consider a \emph{global broadcast} in
which the platform takes one user's private signal and makes the same realization
visible to all other users. Although the signal is private in origin, the broadcast
makes it common in use. It can therefore serve as an additional coordination
device alongside the public signal.

To characterize this effect analytically, we begin with a symmetric large-$N$
benchmark in which all private signals have precision $\beta$. The result below
shows how introducing a broadcast signal changes the role of public information
in equilibrium. We then use Figure~\ref{fig:reversal} to illustrate the same
mechanism in finite networks and to distinguish the effect of common reach from
network centrality alone.

We begin with the no-network benchmark of \citet{morrisshin2002}. Each user
observes only her own private signal $x_i$ and the public signal $y$. In the
unique symmetric linear equilibrium, truth-tracking accuracy is
\begin{equation}\label{eq:WMS}
W_{\mathrm{MS}}(\alpha)
=
-\frac{(1-r)^2\beta+\alpha}
{\big((1-r)\beta+\alpha\big)^2},
\qquad
\frac{dW_{\mathrm{MS}}}{d\alpha}
\ \propto\
\alpha+\beta(1-r)(1-2r).
\end{equation}
When $r>1/2$ and public precision is sufficiently low, this derivative is
negative. Thus, making the public signal more precise can reduce truth-tracking
accuracy because users place too much strategic weight on the common signal as
a coordination device. This is the Morris–Shin transparency effect.

Now suppose the platform amplifies one source $c$ by broadcasting her private signal
$z:=x_c$ to all other users. Each ordinary user observes
$
\{x_i,z,y\}.
$
The broadcast signal is not public information in origin, but it is public-like in use: all
ordinary users observe the same realization. The corresponding accuracy expression is
\begin{equation}\label{eq:WS}
W_{\mathrm{S}}(\alpha)
=
-\frac{\big(1+(1-r)^2\big)\beta+\alpha}
{\big((2-r)\beta+\alpha\big)^2},
\qquad
\frac{dW_{\mathrm{S}}}{d\alpha}
\ \propto\
\alpha+\beta\big(2(1-r)^2+r\big).
\end{equation}
The derivative on the right of \eqref{eq:WS} is positive for every $r\in(0,1)$ and $\alpha>0$.
Broadcasting, therefore, changes the comparative statics of transparency.

\begin{theorem}\label{thm:reversal}
Consider the symmetric large-$N$ benchmark under the truth-tracking accuracy
criterion \eqref{eq:accuracy}.

\begin{enumerate}[label=(\roman*)]

\item \emph{Effective common precision.}
Broadcasting one private signal of precision $\beta$ is equivalent, for
truth-tracking accuracy, to increasing common precision in the no-network
benchmark from $\alpha$ to $\alpha+\beta$:
$
W_{\mathrm{S}}(\alpha)
=
W_{\mathrm{MS}}(\alpha+\beta).
$

\item \emph{Transparency reversal.}
In the no-network benchmark, when $r>\tfrac12$, accuracy decreases with public
precision whenever
$
\alpha<\alpha^\ast
:=
\beta(1-r)(2r-1).
$
By contrast, after one private signal is broadcast,
$\frac{dW_{\mathrm{S}}(\alpha)}{d\alpha}>0$
for every $r\in(0,1)$ and $\alpha>0$. Hence broadcast eliminates the region in
which greater public precision lowers truth-tracking accuracy.

Moreover, whenever $\alpha\ge\alpha^\ast$,
$
W_{\mathrm{S}}(\alpha)
=
W_{\mathrm{MS}}(\alpha+\beta)
>
W_{\mathrm{MS}}(\alpha),
$
broadcasting strictly improves truth-tracking accuracy.
\end{enumerate}
\end{theorem}

Part~(i) of Theorem~\ref{thm:reversal} shows that broadcasting a private signal does not make the original public signal more precise. Instead, it creates an additional signal that is observed by everyone. In the symmetric benchmark, the precision $\beta$ of the broadcast signal therefore enters the accuracy expression exactly as additional common precision.

Part~(ii) shows why this matters for transparency. Without broadcast, the public signal is the only common source of information and can receive excessive strategic weight when coordination motives are strong. Broadcasting introduces a second common signal and moves the economy beyond the region in which greater public precision reduces accuracy. Since $\alpha^\ast\leq\beta/8<\beta$, broadcasting a single private signal is sufficient to cross this threshold. The Morris--Shin non-monotonicity is therefore eliminated because users are no longer forced to coordinate around a single common source of information.

Figure~\ref{fig:reversal} illustrates Theorem~\ref{thm:reversal} in a finite-$N$ environment. Panel~(a) compares broadcast with alternative non-broadcast central positions and shows that the monotonicity result is driven by common reach rather than by centrality alone. Panel~(b) illustrates how information is transmitted under broadcast by plotting the equilibrium weights that a non-central user places on the public signal $y$ and the broadcast signal $z$. The two weights cross at $\alpha=\beta$, showing how users shift toward the relatively more precise source of common information.

\begin{figure}[h!]
\centering
\includegraphics[width=\textwidth]{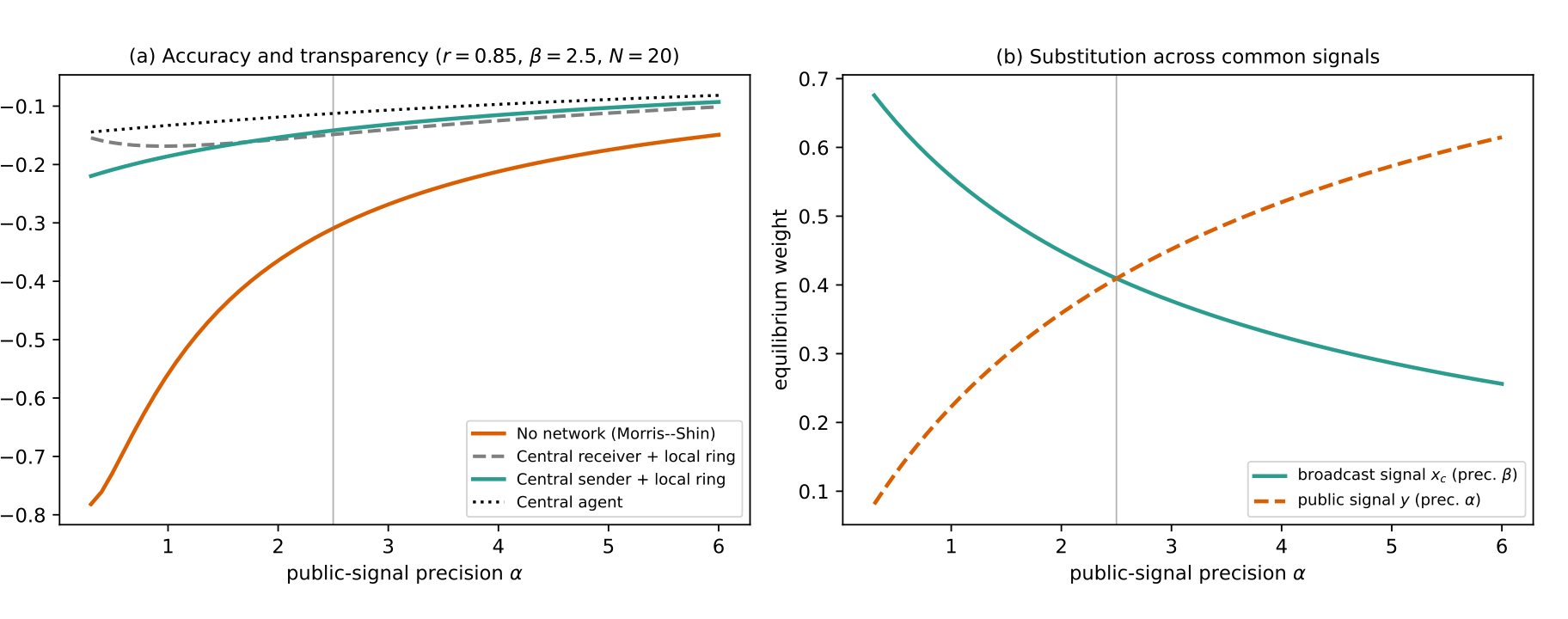}
\caption{Accuracy and equilibrium weights under transparency. \emph{Left:} equilibrium accuracy
as public precision $\alpha$ varies for the no-network benchmark, the central receiver, the
central sender, and the central user ($r=0.85$, $\beta=2.5$, $N=20$). Over the plotted range the
no-network benchmark is increasing; its Morris--Shin downward-sloping region lies at the very low
precisions $\alpha<\beta(1-r)(2r-1)=0.2625$, to the left of the displayed window. The
non-broadcast central-receiver topology retains a visible dip, whereas the broadcast (central
sender and central user) topologies are monotone throughout. \emph{Right:} for a
non-central user, the weight on the broadcast signal $x_c$ and the weight on the public signal
$y$ cross at $\alpha=\beta$, showing the substitution mechanism. All curves are exact finite-$N$
equilibrium objects; the central topologies use the same local follower-ring baseline.}
\label{fig:reversal}
\end{figure}

\paragraph{Optimal amplification under attention constraints.}\label{sec:optimal}

Broadcasting informative signals can improve accuracy, but a platform typically
cannot give every source unlimited reach. We therefore ask which sources should
be amplified when broadcast capacity is limited. We continue to set $b_i=0$ so
that source selection depends only on information quality.
We begin with the unconstrained benchmark.

\begin{lemma}[Bias-free complete-network benchmark]\label{lem:complete}
Suppose $b_i=0$ for all $i$. The complete network maximizes truth-tracking
accuracy. When every user observes every signal, all users choose the
full-information posterior mean; coordination loss is zero, and
\[
W=-\frac{1}{\alpha+\sum_i\beta_i},
\]
which is the highest attainable level of truth-tracking accuracy.
\end{lemma}

 Lemma~\ref{lem:complete} shows why a constraint on reach is necessary for the
source-selection problem to be meaningful. If every signal can be shown to
everyone, the platform simply makes all information available. The interesting
question arises when only a limited number of sources can be given common reach.

To study this problem analytically, consider a symmetric large-$N$ benchmark.
Ordinary users have common private-signal precision $\beta$ and uniform
coordination weights. In addition, there is a finite set of candidate sources
$\mathcal C$, whose signal precisions $\{\beta_\ell\}_{\ell\in\mathcal C}$ may
differ. The platform can broadcast the signals of exactly $m$ candidates. Let
$H\subseteq\mathcal C$, with $|H|=m$, denote the set of amplified sources, and
define
$
T_H=\sum_{\ell\in H}\beta_\ell
$
as their total signal precision.

Signal precisions are exogenous: the platform chooses which sources to amplify,
not how informative those sources are. The candidate set remains finite as
$N\to\infty$, so its share of the population vanishes. Consequently, the
environment of a representative ordinary user depends on the amplified set
$H$ only through the information contained in the broadcast signals.

\begin{theorem}\label{thm:optimal}
In the symmetric large-$N$ broadcast benchmark, truth-tracking accuracy under
an amplified set $H$ is
\begin{equation}\label{eq:Wamp}
W(H)
=
-\frac{(1-r)^2\beta+T_H+\alpha}
{\big((1-r)\beta+T_H+\alpha\big)^2}.
\end{equation}
Hence accuracy depends on the identity of the amplified sources only through
their total precision $T_H$. If $\alpha\geq\beta/8$, accuracy is nondecreasing
in $T_H$, and is strictly increasing except at the knife-edge case
$\alpha=\beta/8$, $T_H=0$, and $r=3/4$. Therefore, with an exact budget of
$m$ broadcast sources, the accuracy-oriented platform amplifies the $m$
highest-precision candidates in $\mathcal C$. In particular, when $m=1$, it
amplifies the most precise source.
\end{theorem}

Theorem~\ref{thm:optimal} gives a simple rule for allocating scarce broadcast
reach. In the monotone region, accuracy depends on the amplified set only
through its total precision $T_H$. Thus, when the platform can broadcast exactly
$m$ sources, it should choose the $m$ most informative ones. Popularity or
existing network centrality plays no role in this benchmark.
This result is related to key-player and targeting problems in network games
\citep{ballester2006,galeotti2020} and to \citet{sakovicssteiner2012}, who study
which agents are most important for coordination. The intervention here is
different: the platform allocates \emph{reach} by deciding whose information should become commonly visible.

The precision ranking, however, relies on the bias-free broadcast benchmark.
Once directional preferences are reintroduced and the network determines both
information exposure and social-reference relationships, source precision alone
is no longer sufficient. Amplifying an informative source can improve learning
while simultaneously increasing the influence of directional biases. We now
turn to this more general trade-off.

\paragraph{Bias and the limits of amplification}\label{sec:bias}  We now reintroduce directional preferences and ask when an informative source can nevertheless reduce truth-tracking accuracy. Private signals remain statistically unbiased, $x_i=\theta+\varepsilon_i$. Directional bias enters only through users' preferred actions, $\theta+b_i$. Thus, a ``biased source'' in our model is not a source that sends distorted information, but a user whose preferred action is shifted away from the state and whose influence may spread through the coordination network.

\begin{proposition}\label{prop:decomp}
In any linear equilibrium, truth-tracking accuracy can be written as
\begin{equation}\label{eq:decomp}
W
=
-\frac{1}{N}\sum_{i=1}^N
\left(
\underbrace{\sum_k\frac{c_{ik}^2}{\beta_k}
+\frac{c_{iy}^2}{\alpha}}_{\text{variance channel }V_i}
+
\underbrace{d_i^2}_{\text{bias channel}}
\right).
\end{equation}
The signal loadings $\{c_{ik},c_{iy}\}$ are independent of the directional
preferences $\{b_i\}$. Hence, bias does not affect the variance channel, including
the transparency result of Theorem~\ref{thm:reversal}; it affects accuracy only
through the equilibrium directional components $\{d_i\}$.
\end{proposition}

 Proposition~\ref{prop:decomp} makes the trade-off transparent. Network
amplification can improve accuracy by reducing information uncertainty, but it
can also decrease accuracy by strengthening the influence of directional
preferences. Before studying arbitrary network interventions, we illustrate this
second force with a simple two-community benchmark. Users are divided into two
equally sized groups with opposite directional preferences. We compare an
integrated network, in which users observe and coordinate with members of both
groups, with two internally complete but segregated communities.

\begin{lemma}\label{lem:polarized-communities}
Suppose $N=2n$ users are divided into two groups of equal size, with a common private-signal precision $\beta$, a common coordination motive $r\in(0,1)$, and biases $b_i\in\{-b,+b\}$, where $b>0$. Compare the following two networks.

\emph{(i) Integrated network.} Every user observes every private signal and coordinates uniformly with all other users.

\emph{(ii) Segregated communities.} Every user observes all private signals generated within her own group and none from the opposite group, and places all coordination weight on members of her own group.

In the integrated network, the equilibrium directional component is
\[
d_i^{I}=\frac{1-r}{1+r/(N-1)}\,b_i,
\]
and truth-tracking accuracy is
\[
W^{I}
=-\frac{1}{\alpha+N\beta}
-\left(\frac{1-r}{1+r/(N-1)}\right)^2 b^2.
\]
In the segregated network, $d_i^{S}=b_i$ and
\[
W^{S}
=-\frac{1}{\alpha+n\beta}-b^2.
\]
Hence $W^{I}>W^{S}$. Even though users in the segregated network observe every signal generated within their own community, integration improves accuracy both by expanding information and by weakening the directional effect of same-sign coordination.
\end{lemma}

Lemma~\ref{lem:polarized-communities} provides a simple interpretation for
community-based platforms. A subreddit-like environment may expose users
densely to content produced within one community while leaving them relatively
separated from users with opposing views. In the model, even complete exposure
within such a community does not eliminate the bias channel: when socially
relevant users share the same directional preference, coordination reinforces
rather than offsets that preference. Dense within-community exposure is
therefore not the same as integrated exposure.\footnote{Financial-investment communities provide a useful example. During
the GameStop episode, investors on Reddit's r/WallStreetBets coordinated around
a common view and trading strategy, even though the resulting stock price was
determined by a much broader set of market participants
\citep{lucchini2022gamestop}. The example illustrates why dense exposure within
one investment community need not eliminate the value of cross-group exposure.
Investors may learn a great deal from other members of the same community, but
their outcomes still depend on the beliefs and actions of investors outside it.
Exposure to those investors can therefore improve truth tracking both by
providing information about the broader market and by making the relevant
coordination group more representative of the agents whose actions determine
market outcomes.}

Proposition~\ref{prop:decomp} captures the broader trade-off. Amplifying an
informative source can reduce uncertainty, but making that source a salient
coordination reference can also increase directional distortion. The platform
must therefore evaluate both the informational benefit of reach and its effect
on the bias channel. The two-community benchmark provides the simplest
illustration of this tension; the results that follow characterize it for
arbitrary finite networks.

To measure the bias effect in an arbitrary network, we use the influence operator from Theorem~\ref{thm:exist}. Define
\[
L(M):=(I-R\Phi(M))^{-1}(I-R),\qquad R=\operatorname{diag}(r_i),
\]
so that the equilibrium bias profile is $d(M)=L(M)b$, and let
\[
\mathcal Q(M):=\frac1N L(M)'L(M).
\]
The matrix $\mathcal Q(M)$ is symmetric and positive semidefinite. Write $L(M)=[\ell_1(M),\ldots,\ell_N(M)]$, where
\[
\ell_c(M)=L(M)e_c=(I-R\Phi(M))^{-1}(I-R)e_c
\]
is the influence profile generated by a unit bias at source $c$. Its diagonal statistic
\[
\mathcal B_c(M):=\mathcal Q_{cc}(M)=\frac1N\|\ell_c(M)\|^2
\]
is the source's \emph{bias-amplification multiplier}. For two distinct sources $c$ and $h$, define
\[
\mathcal C_{ch}(M):=\mathcal Q_{ch}(M)=\frac1N\ell_c(M)'\ell_h(M),
\]
their \emph{influence-overlap coefficient}. Finally, let $\bar V(M)=N^{-1}\sum_i V_i(M)$ denote the mean variance channel and, for an intervention $M_0\to M_1$, write
\[
\Delta V:=\bar V(M_0)-\bar V(M_1),\qquad \Delta\mathcal Q:=\mathcal Q(M_1)-\mathcal Q(M_0).
\]
Here $\Delta V$ is the information gain, and $\Delta\mathcal Q$ records how the intervention changes bias propagation.

\begin{proposition}\label{prop:general-score}
For every finite directed network $M$ and every bias vector $b\in\mathbb R^N$,
\begin{equation}\label{eq:general-bias-accuracy}
W(M)=-\bar V(M)-b'\mathcal Q(M)b.
\end{equation}
Consequently, for any intervention $M_0\to M_1$,
\begin{equation}\label{eq:general-score}
W(M_1)-W(M_0)
=
\underbrace{\Delta V}_{\text{information gain}}
-
\underbrace{b'\Delta\mathcal Q\,b}_{\text{change in bias risk}}.
\end{equation}
\end{proposition}

Next, we show how the network transforms individual directional preferences into aggregate
bias risk
\begin{equation}\label{eq:overlap-decomp}
b'\mathcal Q(M)b
=
\sum_{c=1}^N b_c^2\mathcal B_c(M)
+
2\sum_{1\le c<h\le N}b_cb_h\mathcal C_{ch}(M).
\end{equation}

The diagonal term
$\mathcal B_c(M)$ measures how strongly the bias of source $c$ propagates
through the network. The off-diagonal term $\mathcal C_{ch}(M)$ captures how the
influence of two different sources overlaps. When two sources have biases in
the same direction, their effects reinforce one another; when their biases have
opposite signs, their effects can partially offset.

Proposition~\ref{prop:general-score} therefore gives a simple test for
any network intervention. Moving from $M_0$ to $M_1$ improves truth-tracking
accuracy if and only if the informational benefit of the intervention is large
enough to compensate for the additional bias risk it creates:
$
\Delta V\geq b'\Delta\mathcal Q\,b.
$
Thus, more reach is not valuable simply because it brings additional
information. Whether an intervention is desirable also depends on whose
influence it strengthens and how that influence interacts with the directional
preferences already present in the network. Corollary~\ref{cor:robust} extends
this criterion to the case in which the platform does not know the exact bias
profile but only has a bound on its magnitude.

\begin{corollary}\label{cor:robust}
If $\Delta\mathcal Q\succeq0$ and $\Delta V\ge0$, the bias profiles for which the intervention
raises accuracy form the quadratic region
$\{b:b'\Delta\mathcal Qb\le\Delta V\}$; when
$\Delta\mathcal Q\succ0$ it is an ellipsoid for $\Delta V>0$ and degenerates to $\{0\}$ for
$\Delta V=0$. More generally, for a known bound $\|b\|\le B$, the intervention
raises accuracy for every admissible bias vector if and only if
\begin{equation}\label{eq:robust-score}
\Delta V\ge B^2\max\{\lambda_{\max}(\Delta\mathcal Q),0\}.
\end{equation}
If $\Delta\mathcal Q$ is indefinite, the intervention raises bias risk in some bias directions and
reduces it in others, so no scalar topology-free bias threshold exists without restricting the
composition of biases.
\end{corollary}

Corollary~\ref{cor:robust} shows how much bias an intervention can tolerate before
its informational benefit is overturned. When $\Delta\mathcal Q\succeq0$, the
intervention weakly increases bias risk in every direction. The set
$\{b:b'\Delta\mathcal Q b\leq\Delta V\}$ therefore describes exactly the bias
profiles for which the information gain is large enough to compensate for this
additional distortion. Bias profiles close to zero are safe, while sufficiently
large biases in directions that the intervention amplifies most strongly make
the intervention undesirable.

The second part provides a more conservative criterion when the platform does
not know the exact composition of biases. If it knows only that
$\|b\|\leq B$, the worst-case bias profile is aligned with the eigenvector
associated with the largest positive eigenvalue of $\Delta\mathcal Q$. Hence
$B^2\lambda_{\max}(\Delta\mathcal Q)$ is the largest bias cost the intervention
can generate within the admissible set. The intervention is robustly
accuracy-improving only when its information gain exceeds this worst-case cost.

When $\Delta\mathcal Q$ is indefinite, the effect of the same intervention
depends on the composition of biases in the network. Some bias profiles are
amplified, while others are partially offset. In that case, there is no single
scalar measure of ``how biased'' the population can be that determines whether
the intervention is desirable; the direction of biases and their interaction
with the network matter as well.

Corollary~\ref{cor:robust} provides a robustness criterion when the platform
does not know the exact composition of directional preferences. When the bias
profile $b$ is known, however, the platform can compare candidate sources
directly. The next corollary turns the general intervention criterion into a
source-selection rule.
\begin{corollary}\label{cor:source-ranking}
Fix a baseline network $M$ and let $M^{(c)}$ denote the network obtained by a
candidate one-source amplification of source $c$. Define
\begin{equation}\label{eq:source-score}
\mathcal S_c(M;b)
:=
\bar V(M)-\bar V(M^{(c)})
-
b'\big(\mathcal Q(M^{(c)})-\mathcal Q(M)\big)b.
\end{equation}
Among a set of mutually exclusive one-source interventions, an
accuracy-oriented platform chooses a source that maximizes
$\mathcal S_c(M;b)$. If all scores are strictly negative, no amplification is
optimal; if the largest score is zero, amplification and no intervention are
tied. In general, the score cannot be determined from source $c$'s own precision
and own bias alone. Through the off-diagonal terms in
\eqref{eq:overlap-decomp}, the ranking also depends on the biases and influence
profiles of other users in the network.
\end{corollary}

 Giving source $c$ additional reach creates an
informational benefit, captured by the reduction in $\bar V$, but it can also
change how directional preferences propagate through the coordination network.
The platform should therefore rank sources by the \emph{net} accuracy effect of
amplification rather than by signal precision or bias in isolation. An important implication is that a source's effect depends on the bias
composition already present in the network. A moderately biased source can be
harmful when its influence reinforces other biases in the same direction,
whereas a source with a larger bias can be less harmful when its influence
partially offsets existing distortions. The off-diagonal elements of
$\mathcal Q$ capture precisely these interactions.\footnote{The score in \eqref{eq:source-score} is a marginal comparison against a fixed
baseline network. It therefore ranks mutually exclusive one-source
interventions, but it is not generally additive across several simultaneous
amplifications. If the platform amplifies a set of sources $H$, the joint
network $M^{(H)}$ must be evaluated directly using
\eqref{eq:general-score}, because the influence profiles of the amplified
sources can interact.}

A simple comparison illustrates why the off-diagonal terms matter. Consider two candidate hubs that occupy the same network position and have the same precision and absolute bias $b>0$, but opposite signs: one is aligned with the receiving camp and the other is opposed to it. Because the proposed network intervention is the same, $\Delta V$ and $\Delta\mathcal Q$ are identical across the two candidates, and screening on precision and own bias cannot distinguish them. Changing only the sign of the hub entry gives
\[
\mathcal S_{\mathrm{opposed}}-\mathcal S_{\mathrm{aligned}}
=4b\sum_{h\ne c}b_h\,\Delta\mathcal Q_{ch},
\]
so the diagonal term $b^2\Delta\mathcal B_c$ cancels. The ranking is determined by the cross entries of $\Delta\mathcal Q$ against the ambient bias profile. In the two-camp computation of Section~\ref{sec:engagement}, the proposed amplification creates positive influence overlap with the receiving camp, so the opposed candidate offsets rather than reinforces the existing distortion. At the baseline calibration with $\beta_c=45$ and ordinary-user biases $b_i=\pm0.6$, every relevant $\Delta\mathcal Q_{ch}$ toward the receiving camp is positive, and the decision reverses: $\mathcal S_{\mathrm{aligned}}=-0.0006<0<\mathcal S_{\mathrm{opposed}}=0.0014$. The single-biased-source case considered next removes these interaction terms
and yields a tractable bias threshold.

\begin{proposition}\label{prop:multiplier}
Suppose source $c$ is the only biased user: $b_j=0$ for every $j\ne c$. Then
in any network $M$ the equilibrium bias profile is $d(M)=b_c\,\ell_c(M)$, and the bias channel of
accuracy equals $b_c^2\,\mathcal B_c(M)$, so that
\[
W(M)=-\bar V(M)-b_c^2\,\mathcal B_c(M).
\]
Consequently, for any two networks $M_0$ (source $c$ not amplified) and $M_1$ (source $c$
amplified), writing $\Delta V=\bar V(M_0)-\bar V(M_1)$ for the variance-channel gain and
$\Delta\mathcal B_c=\mathcal B_c(M_1)-\mathcal B_c(M_0)$ for the change in the multiplier,
\[
W(M_1)\ge W(M_0)
\quad\Longleftrightarrow\quad
\Delta V\ \ge\ b_c^2\,\Delta\mathcal B_c .
\]
If amplification raises the source's multiplier, $\Delta\mathcal B_c>0$, a real bias threshold
exists only when the variance gain is nonnegative: for $\Delta V>0$ the condition is
$|b_c|\le b_c^\ast(M_0,M_1):=\sqrt{\Delta V/\Delta\mathcal B_c}$; for $\Delta V=0$ it holds, with
equality, only at $b_c=0$; and for $\Delta V<0$ amplification lowers accuracy at every bias
magnitude. (When instead
$\Delta\mathcal B_c\le0$, amplification weakly raises accuracy for all $b_c$ whenever
$\Delta V\ge0$.) The accuracy comparison itself is exact for every finite directed network. Here the
full bias channel $\mathcal B_c(M)=\frac1N\|\ell_c(M)\|^2$ includes the source's own residual
component $(\ell_c)_c$, so no user's contribution is dropped.
\end{proposition}

Proposition~\ref{prop:multiplier} gives the
information---bias trade-off. The vector $\ell_c(M)$ makes this network effect precise. When only one source is biased, all interaction
terms disappear, and the platform needs to compare just two objects: the
informational gain from giving the source additional reach and the increase in
the source's ability to transmit her directional preference through the
coordination network. The threshold $b_c^\ast$ is the point at which these two
effects exactly offset one another.

This also clarifies that the same value of $b_c$ can be relatively harmless when
the source has little coordination influence but costly when amplification
turns her into an important reference point for many users. Conversely, even a
biased source can be worth amplifying when the information she provides is
sufficiently valuable. This interpretation is related to the role of influential
opinion leaders in networks \citep{loepersteinerstewart2014}.

The next proposition provides a network representation of the bias influence
generated by a source.

\begin{proposition}
\label{prop:bonacich}

For any source $c$, the influence profile $\ell_c(M)$ admits the walk
representation
\[
\ell_c(M)
=
(1-r_c)\sum_{k=0}^{\infty}
\big(R\Phi(M)\big)^k e_c.
\]
Hence, for every user $i$,
\[
(\ell_c(M))_i
=
(1-r_c)\sum_{k=0}^{\infty}
\big[(R\Phi(M))^k\big]_{ic}.
\]

The corresponding bias-amplification multiplier is
\[
\mathcal B_c(M)
=
\frac{(1-r_c)^2}{N}
\sum_{i=1}^N
\left(
\sum_{k=0}^{\infty}
\big[(R\Phi(M))^k\big]_{ic}
\right)^2.
\]
Thus, $\mathcal B_c(M)$ increases with the strength of source $c$'s
equilibrium influence across the coordination network.
\end{proposition}

Proposition~\ref{prop:bonacich} shows how a directional preference propagates
through the network. The term
$\big[(R\Phi(M))^k\big]_{ic}$ captures the influence of source $c$ on user $i$
through coordination paths of length $k$. The direct effect corresponds to
short paths, while longer paths capture indirect influence transmitted through
other users. The equilibrium influence of a source therefore depends not only
on how many users directly treat her as a reference point, but also on the
entire chain of coordination relationships that follows from those links.

This representation connects the bias channel in our model to the familiar
Katz--Bonacich logic of network influence
\citep{katz1953,bonacich1987,ballester2006,jackson2008}. There is, however, an
important distinction. The object $\ell_c(M)$ is a source-specific
\emph{influence profile}: it records how strongly a unit directional preference
at source $c$ affects each user in equilibrium. The multiplier
$\mathcal B_c(M)$ then aggregates the squared magnitude of this profile across
the population. It is therefore not the scalar Bonacich centrality of source
$c$, but a measure of how strongly the source's directional preference is
transmitted throughout the network.

Together with Proposition~\ref{prop:multiplier}, this gives a simple
interpretation of the cost of amplification. Giving a source more reach is
potentially harmful when it strengthens the paths through which her directional
preference affects other users. A source can therefore become costly not merely
because many users observe her, but because amplification places her in a
network position from which her preferred action becomes an important
coordination reference.

Propositions~\ref{prop:multiplier}--\ref{prop:bonacich} are not tied to a particular topology. For any directed network with a single biased source, they identify the source's coordination influence profile as the relevant object. The closed-form threshold derived next is the large-$N$ value of $b_c^*=\sqrt{\Delta V/\Delta\mathcal B_c}$ for one tractable local topology.

For a closed-form illustration, consider a canonical broadcast-with-local-reference benchmark. There is one biased hub with signed bias $b_c$, and representative followers are unbiased. A follower's neighbor-weighted reference group consists of the hub and two representative followers, which produces the factor $3$ in the normalization below. This normalization is specific to the chosen topology and is used only to obtain a tractable benchmark. Let $x$ be the common residual bias of followers and $d_c$ the hub's residual bias. They satisfy
\[
x=r\left[q\frac{d_c+2x}{3}+(1-q)x\right],\qquad d_c=(1-r)b_c+rx.
\]
Hence
\begin{equation}\label{eq:xres}
x
=
\frac{rq(1-r)}{3A-r^2q}\,b_c,\qquad A=1-r+\frac{rq}{3}.
\end{equation}
The follower bias channel grows quadratically in $b_c$. By contrast, in a central-receiver star the hub receives information but is not broadcast, so its bias remains contained. Let
\[
\Delta V=\bar V(M_{\mathrm{receiver}})-\bar V(M_{\mathrm{sender}})>0
\]
be the variance gain from broadcasting. In the large-$N$ limit this is the representative follower's variance gain because the hub's own contribution is $O(1/N)$. Broadcasting improves accuracy relative to the receiver structure exactly when this information gain exceeds the induced bias cost.

\begin{proposition}\label{prop:demagogue}
In the canonical broadcast-with-local-reference benchmark, the demagogue threshold is the
large-$N$ evaluation of the exact multiplier test of Proposition~\ref{prop:multiplier} in this
local topology. To leading order in $N$, amplifying the hub raises accuracy relative to not
amplifying it if and only if
\[
|b_c|\le b_c^\ast,
\]
where
\begin{equation}\label{eq:threshold}
b_c^\ast
=
\sqrt{\Delta V}\,\frac{3A-r^2q}{rq(1-r)}
=
\sqrt{\Delta V}\,\frac{3+rq}{rq},
\qquad
A=1-r+\frac{rq}{3}.
\end{equation}
In this canonical topology, the tolerable bias rises with the informational gain $\Delta V$ and
falls with amplification intensity $q$ and, holding $\Delta V$ fixed, with the coordination motive
$r$. The threshold should be read as a benchmark formula—the large-$N$, leading-order value of
the exact finite-$N$ threshold $b_c^\ast(M_0,M_1)=\sqrt{\Delta V/\Delta\mathcal B_c}$ of
Proposition~\ref{prop:multiplier}—not as a topology-invariant constant.
\end{proposition}

The approximation closely tracks the exact finite-$N$ crossover in the calibration used below.\footnote{Appendix~\ref{app:bias-validation} reports the finite-$N$ crossovers, the corresponding formula values, and convergence diagnostics.}

Proposition~\ref{prop:demagogue} gives a leading-order threshold, while Figure~\ref{fig:demagogue} reports the exact finite-$N$ comparison in the same local-ring topology. Followers observe two local neighbors; the central receiver observes the followers but is not observed by them, whereas the central sender is observed by the followers and becomes part of their local reference group. Because a follower's neighbor-weighted reference group under the sender includes the hub plus two ring neighbors, the figure uses the same local structure that generates the factor $3$ in \eqref{eq:threshold}. Solving the finite-$N$ equilibrium gives the receiver an accuracy of approximately $-0.147$ at $b_c=0$, and the sender crosses below the receiver at $b_c\approx0.41$. When the hub is neutral, the central sender performs best because it spreads precise information. As $b_c$ rises, the same reach also spreads directional distortion through the reference channel. Beyond the threshold, broadcasting the hub is no longer accuracy-improving relative to the receiver benchmark. This is the demagogue effect.

\begin{figure}[t]
\centering
\includegraphics[width=\textwidth]{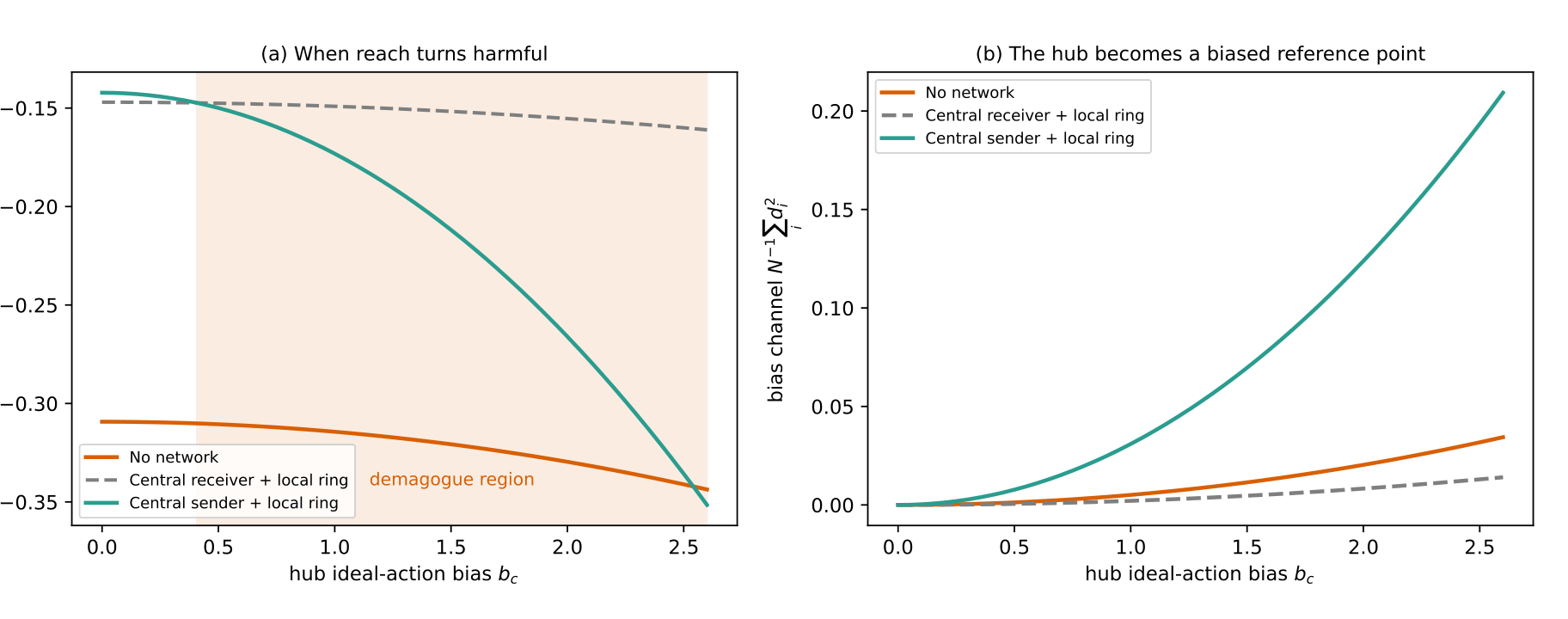}
\caption{The demagogue effect in the finite-$N$ local-ring topology ($r=0.85$, $q=0.7$,
$\beta=\alpha=2.5$, $N=16$). \emph{Left:} accuracy as the hub's own bias $b_c$ increases. The
central sender is best when the hub is neutral; at the demagogue threshold it falls below the
central-receiver benchmark, and at larger bias it becomes the worst of the plotted structures.
\emph{Right:} the polarization channel $N^{-1}\sum_i d_i^2$ rises quadratically under the central
sender, which makes the hub a biased coordination reference, while the central receiver keeps that
bias contained. Because the followers form a common local ring, the central receiver is not the
no-network benchmark.}
\label{fig:demagogue}
\end{figure}

\section{Implementation with User Choice}\label{sec:implementation}

So far, the platform has been able to choose the network structure directly in
order to improve truth-tracking accuracy. We now relax this assumption. The
platform can recommend changes in users' feeds, but users decide whether those
changes are privately desirable.

For a given network $M$, let
$
U_i(M)
:=
\E\!\left[
u_i\big(a^\ast(M),\theta\big)
\right]
$
denote user $i$'s equilibrium expected utility, where $a^\ast(M)$ is the
equilibrium action profile induced by $M$. Suppose the platform considers a
change from $M$ to $M'$. The platform recommends the change whenever
$
W(M')>W(M),
$
whereas user $i$ accepts it only if
$
U_i(M')>U_i(M).
$
The implementation problem arises when these two rankings disagree: a link may
improve collective truth tracking while being privately unattractive to the
user who must adopt it.

We have in mind a directed follower network, such as a Twitter-like platform.
A user can follow another source without requiring that source to follow her
back or approve the connection. Link formation is therefore unilateral rather
than bilateral. Likewise, a source can disappear from a user's effective feed
without requiring the source's consent. A deletion in the model need not be
interpreted literally as the platform recommending an ``unfollow.'' It can
instead represent a change in feed curation that makes disagreement more
salient or reduces the attractiveness of maintaining that source in the
user's feed, inducing the user to drop it herself.\footnote{For this reason,
the rest points below are not pairwise-stable networks in the bilateral
network-formation sense of \citet{jacksonwolinsky1996}.}

This distinction is particularly relevant for cross-cutting links. A precise
source with an opposing directional preference may improve collective accuracy
by providing valuable information or reducing the concentration of
same-direction influence, while at the same time lowering the receiving
user's expected utility by making her coordination environment less attractive.
Thus, the network preferred by an accuracy-oriented platform need not be
privately implementable.

We study this tension through an iterative recommendation procedure. At each
update round, the platform evaluates candidate link additions and deletions and
recommends those that strictly increase $W$. Users then accept or reject these
recommendations according to their own expected utility. Accepted changes
generate the next network,
$
M_{t+1}=\Psi(M_t).
$
A network is a rest point when none of the accuracy-improving changes proposed
by the platform is privately acceptable to the relevant user.

The procedure is not intended as a model of forward-looking network formation:
users do not solve an intertemporal problem. Rather, it provides a simple way
to study whether the accuracy-improving network interventions characterized
above can be implemented when users retain control over their own feeds.
Appendix~\ref{app:sim} gives the complete algorithm and computational details.
Section~\ref{sec:engagement} keeps this recommendation environment fixed but
changes the platform's objective, allowing us to compare an accuracy-oriented
recommender with an engagement-oriented one.

\paragraph{Private Incentives and Platform Recommendations.}

The recommendation constraint is almost irrelevant when users have no
directional preferences. In the bias-free benchmark, a link that improves
collective accuracy is usually also attractive to the user who receives it:
users accept $97\%$ of the platform's accuracy-improving recommendations.

This changes sharply once users have directional preferences. With two equally
sized camps, $b_i\in\{-0.6,+0.6\}$, users accept only $32\%$ of the same type
of accuracy-improving recommendations. The rejected links are predominantly
cross-cutting: $74\%$ involve a source with the opposite sign. Thus, the main
obstacle to implementing the accuracy-improving network is not user choice
itself, but users' reluctance to connect to informative sources whose preferred
direction differs from their own.

The reason is that users and the platform value cross-cutting links
differently. Such a link can improve $W$ by adding information and reducing 
the concentration of same-direction influence. The receiving user, however,
also cares about coordinating with the users who become salient in her feed.
An informative opposite-sign source can therefore be socially valuable while
remaining privately unattractive. User choice consequently acts as a selection
force against precisely some of the links that an accuracy-oriented platform
would like to create.

Figure~\ref{fig:formation} shows the network consequences of this wedge. We
compare the accuracy-oriented recommender with two decentralized benchmarks:
friends-of-friends recommendations and random rewiring. At both $N=20$ and
$N=50$, the accuracy-oriented rule produces higher truth-tracking accuracy.
Friends-of-friends recommendations instead reinforce homophily, while random
rewiring expands exposure without systematically directing attention toward
informative sources. When directional preferences are removed, these
differences largely disappear, indicating that the ranking is driven primarily
by the bias channel rather than by network density alone.

\begin{figure}[t]
\centering
\includegraphics[width=\textwidth]{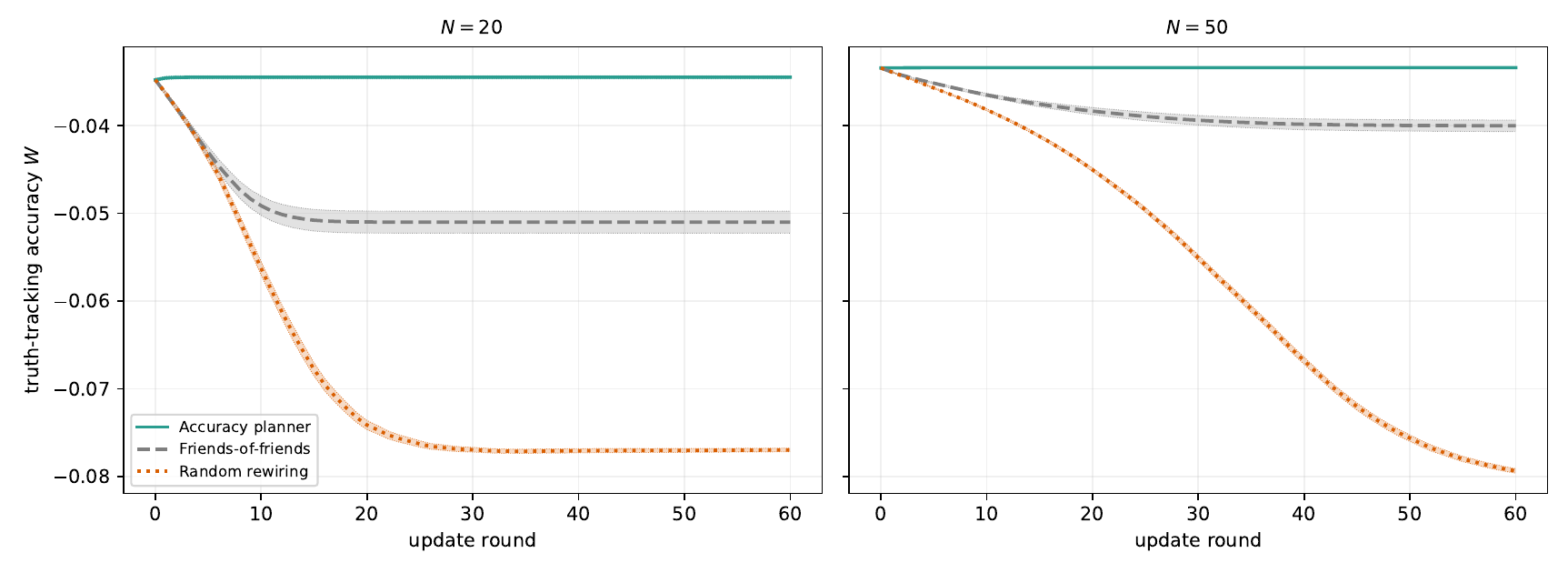}
\caption{Recommendation dynamics at $N=20$ and $N=50$ over $60$ update rounds
($r=0.7$, $q=0.5$, $\alpha=30$, $b_i\in\{-0.6,+0.6\}$; 100 paired Monte
Carlo runs at each size). Lines are means and shaded regions are pointwise
95\% confidence intervals. The accuracy-oriented recommender maintains higher
truth-tracking accuracy than friends-of-friends and random rewiring at both
network sizes. Appendix~\ref{app:sim} reports the simulation protocol and
robustness exercises.}
\label{fig:formation}
\end{figure}

The broader implication is that optimal network design and implementable
network design need not coincide. Even when the platform knows which links
would improve collective accuracy, those links may fail to form because the
users who must adopt them evaluate the same intervention according to their own
preferences. The implementation constraint is therefore endogenous to user
choice rather than a consequence of the platform's inability to identify the
accuracy-improving network.

\section{When the Platform Values Engagement}\label{sec:engagement}

The previous section keeps the platform's objective fixed and limits its ability
to implement the desired network. We now consider a different departure from
the benchmark: the platform itself may not value truth-tracking accuracy.
Commercial recommendation systems may instead place weight on engagement,
retention, or other measures of user activity
\citep{zhuravskaya2020,aridor2024}. We therefore ask how the allocation of
reach changes when the platform values outcomes that users themselves find
attractive.

Our purpose is not to build a structural model of clicks or platform revenue.
Instead, we use average expected user utility as a reduced-form engagement
criterion. This isolates a tension already present in the model. Accuracy
rewards actions that track the state, whereas users prefer actions closer to
their own directional ideals and also value coordination with socially salient
neighbors.

\paragraph{Accuracy and engagement.}

Suppose that a reduced-form retention index is increasing in expected utility,
$
R_i=\rho_0+\rho_1\E[u_i(a,\theta)] $, $\rho_1>0.
$
Maximizing average retention is then equivalent to maximizing
\begin{equation}\label{eq:engagement}
\begin{split}
\Pi
&=\frac1N\sum_{i=1}^N \E[u_i(a,\theta)]\\
&=-\frac1N\sum_i
\left[
(1-r_i)\big(V_i+(d_i-b_i)^2\big)
+r_i\sum_j\phi_{ij}
\big(V_i+V_j-2\Cov_{ij}+(d_j-d_i)^2\big)
\right].
\end{split}
\end{equation}

The two objectives evaluate the same equilibrium differently. The
truth-tracking criterion $W$ penalizes deviations of actions from the state
$\theta$. The engagement criterion $\Pi$ instead evaluates actions relative to
users' preferred points $\theta+b_i$ and also rewards coordination with salient
neighbors. A network intervention can therefore raise users' expected utility
even when it moves equilibrium actions farther from the state.

\begin{proposition}\label{prop:divergence}
Holding signal loadings, variances, covariances, and reference weights fixed,
the two objectives evaluate the directional component differently.
Truth-tracking accuracy is maximized at $d_i=0$ for every user. By contrast,
the engagement objective rewards movements of $d_i$ toward the user's own
ideal $b_i$ and movements that reduce disagreement among socially connected
users. The engagement-optimal directional components, therefore, need not
coincide with the accuracy-optimal ones.
\end{proposition}

Proposition~\ref{prop:divergence} identifies the source of the divergence. An
engagement-oriented platform may value a same-sign connection even when the
source contributes little information because the connection moves actions
toward users' preferred directions and facilitates coordination within their
reference groups. Engagement incentives can therefore favor more homogeneous
local environments, even when those environments are less attractive from the
perspective of truth tracking.

Figure~\ref{fig:engagement} shows that this distinction survives once the full
equilibrium response to network changes is taken into account. In the baseline
simulations, each recommender performs better according to its own objective:
the accuracy-oriented platform generates higher $W$, while the
engagement-oriented platform generates higher $\Pi$. The engagement-oriented
network also exhibits greater polarization and larger cross-camp action gaps at
both $N=20$ and $N=50$.

\begin{figure}[t]
\centering
\includegraphics[width=\textwidth]{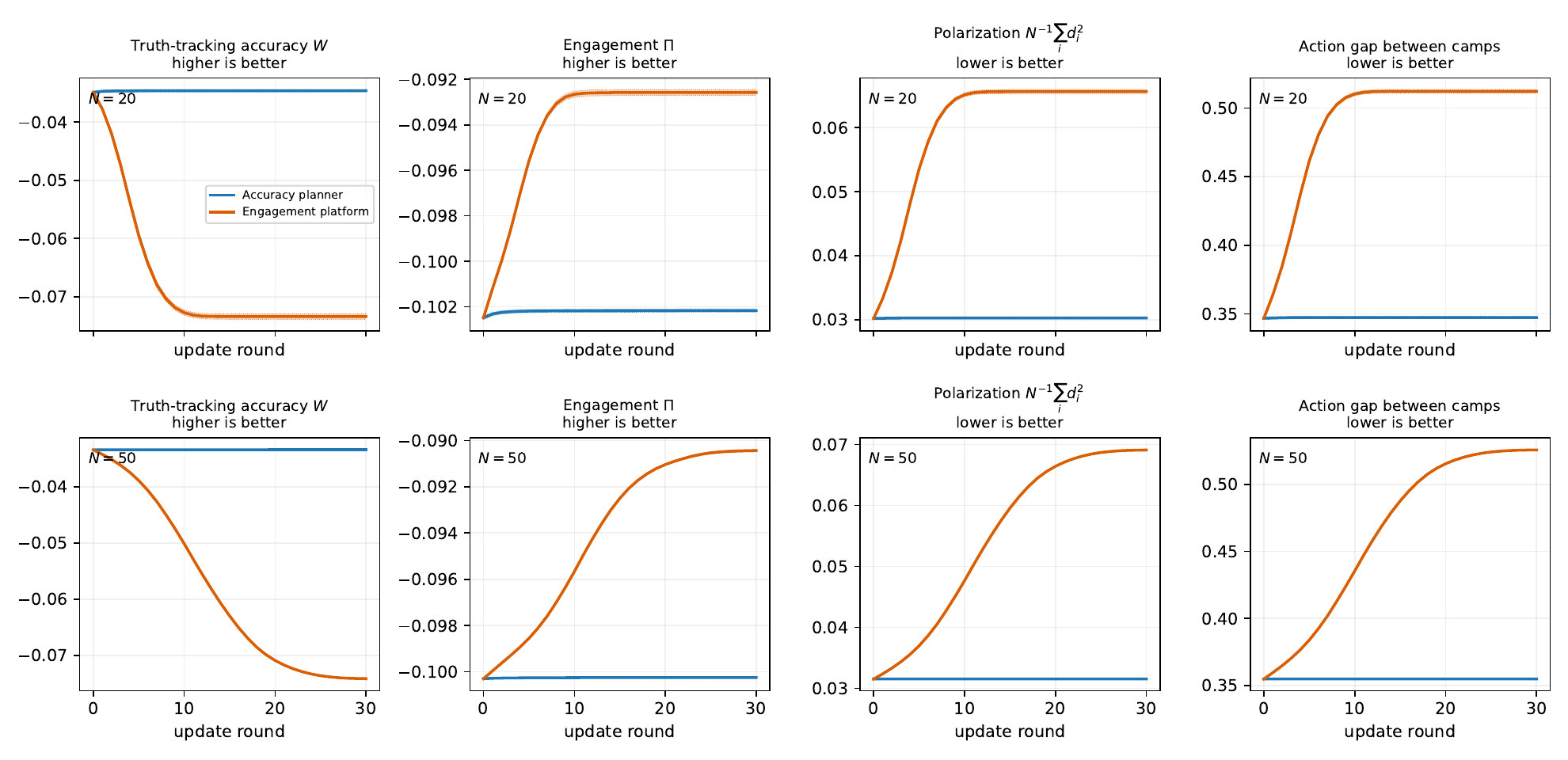}
\caption{Truth versus engagement at $N=20$ (top row) and $N=50$ (bottom row)
under the two platform objectives ($r=0.7$, $q=0.5$, $\alpha=30$,
$b_i\in\{-0.6,+0.6\}$; 100 paired Monte Carlo runs at each size).
The engagement-oriented recommender attains higher $\Pi$ but lower
truth-tracking accuracy, greater polarization, and a larger cross-camp
residual-action gap. Appendix~\ref{app:sim} reports the numerical values and
additional robustness exercises.}
\label{fig:engagement}
\end{figure}

Proposition~\ref{prop:divergence} shows why accuracy and engagement can value
the same network change differently. We now ask whether this difference is
large enough to reverse an actual amplification decision.

Consider an intervention from $M_0$ to $M_1$. Its effect on truth-tracking
accuracy is
\[
\Delta W
=
\underbrace{\Delta V}_{\text{information gain}}
-
\underbrace{b'\Delta\mathcal Q b}_{\text{change in bias risk}}.
\]
An accuracy-oriented platform therefore favors amplification only when its
informational benefit is large enough to compensate for the additional
directional distortion it creates. The engagement-oriented platform evaluates
the same intervention differently: it also values movements toward users'
preferred actions and reductions in disagreement within their reference
groups.

This difference generates a simple pattern  as the
amplified source becomes increasingly biased:
\[
\begin{array}{lll}
\text{informative, moderately biased source}
&:&
\Delta W>0,\quad \Delta\Pi>0,
\\[1.5mm]
\text{validating source, but too biased for accuracy}
&:&
\Delta W<0,\quad \Delta\Pi>0,
\\[1.5mm]
\text{very extreme source}
&:&
\Delta W<0,\quad \Delta\Pi<0.
\end{array}
\]

The first region creates no conflict: the source provides enough information
to improve truth tracking and is also attractive from the perspective of user
utility. In the second region, the information gain is no longer sufficient
to offset the source's bias cost, so the accuracy-oriented platform rejects
amplification. Yet users still benefit from validation and coordination with
like-minded reference groups, so the engagement-oriented platform accepts it.
This is the region in which the two platform objectives disagree:
$
\Delta W<0<\Delta\Pi.
$
If the source becomes sufficiently extreme, however, validation eventually
gives way to overshooting or additional coordination losses, and engagement
also falls. The disagreement is therefore an intermediate region rather than
a mechanical consequence of source bias.\footnote{ Appendix~\ref{app:sim} shows that
this three-region pattern persists across alternative coordination parameters
and bias magnitudes.}

The same trade-off can be expressed directly in the platform's objective.
Suppose the platform places weight on both truth-tracking accuracy and
engagement:
\[
O_\lambda=(1-\lambda)W+\lambda\Pi,
\qquad \lambda\in[0,1].
\]
For a fixed intervention,
\[
\Delta O_\lambda
=
(1-\lambda)\Delta W+\lambda\Delta\Pi.
\]
Consider the disagreement region, where $\Delta W<0<\Delta\Pi$. When
$\lambda$ is small, the accuracy loss dominates and the platform rejects the
intervention. As the weight on engagement increases, the same intervention
eventually becomes attractive. The switch occurs at
$
\lambda^*
=
\frac{-\Delta W}{\Delta\Pi-\Delta W},
$
so that
\[
\Delta O_\lambda
\begin{cases}
<0, & \lambda<\lambda^*,\\
=0, & \lambda=\lambda^*,\\
>0, & \lambda>\lambda^*.
\end{cases}
\]

The cutoff $\lambda^*$ tells us how much weight the platform must place on
engagement before it changes its decision about a particular source. If
$\lambda<\lambda^*$, the loss in truth-tracking accuracy dominates and the
platform does not amplify the source. If $\lambda>\lambda^*$, the engagement
gain dominates and the platform amplifies it. Because the accuracy loss and
engagement gain differ across sources and network interventions, $\lambda^*$
also differs across cases. It should therefore be interpreted as a
source-specific switching point, rather than as a universal regulatory
threshold.

\paragraph{The role of social-reference coupling.}\label{sec:decouple}
The baseline model gives a feed link two roles: it determines whose information
a user observes and who becomes salient for coordination. If the coordination
matrix $\Phi$ is held fixed while the information network changes, the
directional component $d=(I-R\Phi)^{-1}(I-R)b$ is unchanged. Purely
informational rewiring can therefore improve information aggregation, but it
cannot change polarization through the bias-propagation channel studied here.
Thus, the broadcast and precision-ranking results do not rely on this coupling,
whereas the polarization effect arises because amplification also changes who
becomes a salient reference point. Appendix~\ref{app:sim} verifies this
invariance numerically.






\section{Conclusion}\label{sec:discussion}

This paper studies the allocation of reach as a network-design problem. A platform chooses who observes whom in an environment where agents care both about tracking an underlying state and coordinating with others. The central difficulty is that a broadcast  link can have two effects at once: it can provide information, and it can make the source more important as a coordination reference. The value of amplification, therefore, depends not only on what a source knows but also on how the source's preferred action propagates through the network.

The main theoretical result gives an exact characterization of this trade-off for arbitrary finite directed networks. Equilibrium separates a bias-independent information component from a precision-independent directional component. As a consequence, the change in truth-tracking accuracy generated by any network intervention equals its information gain minus the change in quadratic bias risk. The influence-overlap matrix summarizes this second term. Its diagonal entries capture the propagation of individual sources' directional preferences, while its off-diagonal entries capture reinforcement and offsetting between sources. This is why neither precision nor own bias alone is sufficient for ranking amplification decisions: what matters is how a candidate source's influence interacts with the bias composition already present in the network.

The benchmark results illustrate different margins of this general criterion. In the bias-free broadcast environment, making a private signal commonly visible creates a second common signal and eliminates the Morris–Shin region in which greater public precision lowers truth-tracking accuracy. When broadcast capacity is scarce and accuracy is monotone in total broadcast precision, the optimal rule is to allocate reach to the most precise sources. Once directional preferences affect coordination, however, the same informational intervention can become harmful. The single-biased-source case yields a threshold at which the information gain is exactly offset by the source's network-amplified bias cost, while the two-community benchmark shows that complete exposure within polarized groups can remain inferior to integrated exposure.

Implementation introduces a distinct wedge between desirable and privately acceptable networks. In the recommendation simulations, users are willing to accept most accuracy-improving links when directional preferences are absent but frequently reject informative cross-cutting links when they coordinate around same-sign reference groups. An engagement-oriented recommender strengthens this force because expected user utility rewards both validation and coordination comfort. In the baseline simulations, the resulting networks are more segregated and less accurate than those generated by the accuracy-oriented planner. Decoupling information exposure from coordination salience eliminates the polarization difference, identifying the reference-group channel as the mechanism behind this result.

The analysis suggests a broader distinction for information design in networks. Increasing exposure and increasing influence are not the same intervention. A designer may want to deliver information from a precise source without simultaneously making that source a dominant coordination reference. The model abstracts from endogenous signal production, strategic content choice, and forward-looking learning, all of which could interact with reach. Introducing these margins would change the incentives to produce and transmit information, but the basic trade-off would remain: network design determines both what agents learn and, when exposure creates social salience, whose preferences shape the actions of others.

\begin{singlespace}

\end{singlespace}

\newpage
\section*{Appendix}

\appendix

\section{Proof of Theorem~\ref{thm:exist}}\label{app:exist}

We first make the strategy space precise. For a fixed network $M$, a strategy of user $i$ is a measurable function of $y$, $x_i$, and $(x_j)_{j\in\Ni(M)}$. A strategy is \emph{translation-equivariant} if, for every common shift $t$,
\[
g_i\big(y+t,x_i+t,(x_j+t)_{j\in\Ni(M)}\big)
=
g_i\big(y,x_i,(x_j)_{j\in\Ni(M)}\big)+t.
\]
Let $\mathcal S_i(M)$ denote the set of translation-equivariant strategies satisfying $\E[(a_i-\theta)^2]<\infty$, and let $\mathcal S(M)=\prod_i\mathcal S_i(M)$. Expectations are taken over the signal noises. This formulation makes tracking errors square-integrable under the diffuse prior.

The proof has three parts. We first establish uniqueness for the diffuse-prior game, then construct the equilibrium as the limit of proper-prior equilibria, and finally show that the signal and bias coefficients solve separate systems.

\paragraph{Step 1: uniqueness.}
Suppose $a$ and $\widetilde a$ are two equilibria in $\mathcal S(M)$, and define
$
\delta_i=a_i-\widetilde a_i.
$
Because both strategies are translation-equivariant and have finite tracking error, $\delta_i$ is square-integrable. Subtracting the two best-reply equations \eqref{eq:bestreply} eliminates both the posterior term and the bias term:
\begin{equation}\label{eq:delta-fixed-point}
\delta_i
=
r_i\sum_{j\ne i}\phi_{ij}(M)\E_i[\delta_j].
\end{equation}
Conditional expectation is a contraction in $L^2$, so
$
\big\|\E_i[\delta_j]\big\|_2\le \|\delta_j\|_2.
$
Using $\sum_{j\ne i}\phi_{ij}(M)=1$ in \eqref{eq:delta-fixed-point},
$
\|\delta_i\|_2
\le
r_i\max_j\|\delta_j\|_2
\le
\bar r\max_j\|\delta_j\|_2.
$
Taking the maximum over $i$ gives
$
\max_i\|\delta_i\|_2
\le
\bar r\max_i\|\delta_i\|_2.
$
Since $\bar r<1$, the only possibility is $\delta_i=0$ almost surely for every $i$. Hence an equilibrium in $\mathcal S(M)$ is unique.

\paragraph{Step 2: existence and linearity.}
To construct the equilibrium, replace the diffuse prior by the proper prior
$
\theta\sim\mathcal N(0,\tau^{-1}),\qquad \tau>0.
$
For each $\tau$, define the best-reply operator on the product of users' square-integrable strategy spaces by
$
(T_\tau a)_i
=
(1-r_i)\big(\E_i^\tau[\theta]+b_i\big)
+r_i\sum_{j\ne i}\phi_{ij}(M)\E_i^\tau[a_j].
$
For two strategy profiles $a$ and $\widetilde a$, the same argument as above yields
$
\|T_\tau a-T_\tau\widetilde a\|
\le \bar r\|a-\widetilde a\|,
$
where the norm is the maximum of the users' $L^2$ norms. Thus $T_\tau$ is a contraction, so the proper-prior game has a unique Bayes--Nash equilibrium.

Because the state and signals are jointly Gaussian, conditional expectations of affine functions are affine. The best-reply operator therefore maps affine strategies into affine strategies. Iterating $T_\tau$ from an affine profile converges to the unique fixed point, which must consequently be affine. Write this equilibrium as
$
a_i^\tau=\sum_{k=1}^N c_{ik}^\tau x_k+c_{iy}^\tau y+d_i^\tau,
$
with $c_{ik}^\tau=0$ when user $i$ does not observe $x_k$.

The affine coefficient system is finite-dimensional and depends continuously on $\tau$. At $\tau=0$, the homogeneous system associated with two candidate affine equilibria has only the zero solution by Step 1, so its coefficient matrix is nonsingular. Hence the proper-prior coefficients converge as $\tau\downarrow0$. Their limit satisfies the diffuse-prior best-reply equations and therefore provides existence. By Step 1, this limit is the unique equilibrium in $\mathcal S(M)$.

\paragraph{Step 3: separation of signal loadings and biases.}
Under the diffuse prior, let
$
D_i=\alpha+\beta_i+\sum_{h\in\Ni(M)}\beta_h
$
be user $i$'s posterior precision. Matching coefficients in \eqref{eq:bestreply} gives, for every private signal $x_k$ observed by $i$,
\begin{equation}\label{eq:loading-system-app}
c_{ik}
=
A_i\frac{\beta_k}{D_i}
+r_i\sum_{j\ne i}\phi_{ij}(M)c_{jk},
\end{equation}
and for the public signal,
\begin{equation}\label{eq:loading-system-y-app}
c_{iy}
=
A_i\frac{\alpha}{D_i}
+r_i\sum_{j\ne i}\phi_{ij}(M)c_{jy},
\end{equation}
where
$
A_i
=(1-r_i)
+r_i\sum_{j\ne i}\phi_{ij}(M)
\sum_{k\notin\{i\}\cup\Ni(M)}c_{jk}.
$
The term $A_i$ collects the direct accuracy motive and the part of user $i$'s higher-order forecast that loads on signals she does not observe. Crucially, $b$ does not appear in \eqref{eq:loading-system-app}--\eqref{eq:loading-system-y-app}. The equilibrium signal loadings are therefore independent of directional biases.
The loadings add up to one. Let
$
s_i=\sum_k c_{ik}+c_{iy}.
$
Summing the coefficient equations yields
$
s_i=(1-r_i)+r_i\sum_{j\ne i}\phi_{ij}(M)s_j.
$
The vector $s=\mathbf 1$ solves this system. Since $\rho(R\Phi(M))\le\bar r<1$, the solution is unique, and therefore
\begin{equation}\label{eq:loadings-sum-one}
\sum_k c_{ik}+c_{iy}=1
\qquad\text{for every }i.
\end{equation}

Finally, matching the constant terms gives
$
d_i=(1-r_i)b_i+r_i\sum_{j\ne i}\phi_{ij}(M)d_j,
$
or in matrix form,
$
(I-R\Phi(M))d=(I-R)b.
$
Since $\rho(R\Phi(M))<1$, the inverse exists and
$
d=(I-R\Phi(M))^{-1}(I-R)b.
$
This system contains no signal precision. Hence, the constants are independent of $(\alpha,\{\beta_i\})$, while the signal loadings are independent of $b$. This proves the theorem.
\qed

\section{Proof of Theorem~\ref{thm:reversal}: equilibrium weights and accuracy}\label{app:reversal}

\paragraph{No-network benchmark.}
A representative user observes only $x_i$ and $y$. By symmetry, write
$
a_i=w_xx_i+w_yy,\qquad w_x+w_y=1.
$
The posterior mean is
$
\E_i[\theta]=\frac{\beta x_i+\alpha y}{\beta+\alpha}.
$
In the large-$N$ limit, user $i$ expects the population action to be
$
\E_i[\bar a]=w_x\E_i[\theta]+w_yy.
$
Substituting into the best reply gives
$
a_i=K\E_i[\theta]+rw_yy
$, $ K=(1-r)+rw_x.
$
Matching the coefficients on $x_i$ and $y$ yields
$
w_x=\frac{(1-r)\beta}{(1-r)\beta+\alpha},
\qquad
w_y=\frac{\alpha}{(1-r)\beta+\alpha}.
$
Because the signal noises are independent, the mean-squared tracking error is
$
\frac{w_x^2}{\beta}+\frac{w_y^2}{\alpha}
=
\frac{(1-r)^2\beta+\alpha}{((1-r)\beta+\alpha)^2},
$
which gives \eqref{eq:WMS}.

\paragraph{Broadcast benchmark.}
Now let $z=x_c$ be broadcast. A representative non-amplified user observes $(x_i,z,y)$. The source $c$ herself has vanishing population share, so her different information set does not affect population-average accuracy in the large-$N$ limit. Write
$
a_i=w_xx_i+w_zz+w_yy.
$
The posterior mean is
$
\E_i[\theta]=\frac{\beta x_i+\beta z+\alpha y}{2\beta+\alpha},
$
and, because $z$ and $y$ are common,
$
\E_i[\bar a]=w_x\E_i[\theta]+w_zz+w_yy.
$
The best reply becomes
$
a_i=K\E_i[\theta]+rw_zz+rw_yy,
\qquad K=(1-r)+rw_x.
$
Matching coefficients gives
$
w_x=\frac{(1-r)\beta}{(2-r)\beta+\alpha},\qquad
w_z=\frac{\beta}{(2-r)\beta+\alpha},\qquad
w_y=\frac{\alpha}{(2-r)\beta+\alpha}.
$
These weights sum to one. The corresponding mean-squared error is
$
\frac{w_x^2}{\beta}+\frac{w_z^2}{\beta}+\frac{w_y^2}{\alpha}
=
\frac{(1+(1-r)^2)\beta+\alpha}{((2-r)\beta+\alpha)^2},
$
which gives \eqref{eq:WS}.

\paragraph{Effective common precision.}
Evaluate the no-network expression at $\alpha+\beta$:
$
W_{\mathrm{MS}}(\alpha+\beta)
=
-\frac{(1-r)^2\beta+\alpha+\beta}
{((1-r)\beta+\alpha+\beta)^2}
=
W_{\mathrm S}(\alpha).
$
This proves part (i).

\paragraph{Derivative signs.}
Differentiating \eqref{eq:WMS} gives
$
\operatorname{sign}\left(\frac{dW_{\mathrm{MS}}}{d\alpha}\right)
=
\operatorname{sign}\big(\alpha+\beta(1-r)(1-2r)\big).
$
Thus, for $r>1/2$, the derivative is negative exactly when
$
\alpha<\beta(1-r)(2r-1)=\alpha^*.
$
The function $(1-r)(2r-1)$ reaches its maximum $1/8$ at $r=3/4$, so $\alpha^*\le\beta/8<\beta$. Broadcasting shifts effective common precision from $\alpha$ to $\alpha+\beta$, which lies strictly above the non-monotone region. Hence $dW_{\mathrm S}/d\alpha>0$ for every $r\in(0,1)$ and $\alpha>0$. The comparison of levels in the monotone region follows immediately from
$W_{\mathrm S}(\alpha)=W_{\mathrm{MS}}(\alpha+\beta)$. This proves part (ii).
\qed

\section{Proofs for Section~\ref{sec:optimal}}\label{app:optimal}

\begin{proof}[Proof of Lemma~\ref{lem:complete}]
In the complete network every user observes the public signal and every private signal. Hence all users have the same full-information posterior mean,
$
\widehat\theta
=
\frac{\alpha y+\sum_{i=1}^N\beta_i x_i}
{\alpha+\sum_{i=1}^N\beta_i}.
$
If all users choose $a_i=\widehat\theta$, the coordination loss is zero. The tracking mean-squared error is the posterior variance
$
\Var(\theta\mid y,x_1,\ldots,x_N)
=
\frac{1}{\alpha+\sum_i\beta_i}.
$
No alternative network can provide any user with more information than the full collection $(y,x_1,\ldots,x_N)$, and the posterior mean minimizes mean-squared error among actions based on that information. Hence no network can achieve a lower average tracking error. Therefore the complete network is accuracy-optimal and
$
W=-\frac{1}{\alpha+\sum_i\beta_i}.
$
\end{proof}

\begin{proof}[Proof of Theorem~\ref{thm:optimal}]
Fix an amplified set $H$ of size $m$ and let
$
T_H=\sum_{\ell\in H}\beta_\ell.
$
The candidate set is finite and $m/N\to0$, so population-average accuracy is governed by a representative ordinary user. Let the amplified signals be $z_1,\ldots,z_m$, with precisions $\beta_1,\ldots,\beta_m$. The representative user observes $(x_i,z_1,\ldots,z_m,y)$, where $x_i$ has precision $\beta$.

The posterior mean is
$
\E_i[\theta]
=
\frac{\beta x_i+\sum_{\ell\in H}\beta_\ell z_\ell+\alpha y}
{\beta+T_H+\alpha}.
$
Write
$
a_i=w_xx_i+\sum_{\ell\in H}w_{z_\ell}z_\ell+w_yy.
$
Because the amplified signals and $y$ are common, the same coefficient matching as in the broadcast proof gives, with
$
D=(1-r)\beta+T_H+\alpha,
$
$
w_x=\frac{(1-r)\beta}{D},\qquad
w_{z_\ell}=\frac{\beta_\ell}{D},\qquad
w_y=\frac{\alpha}{D}.
$
The weights sum to one. Therefore the tracking mean-squared error is
$
\frac{w_x^2}{\beta}
+\sum_{\ell\in H}\frac{w_{z_\ell}^2}{\beta_\ell}
+\frac{w_y^2}{\alpha}
=
\frac{(1-r)^2\beta+T_H+\alpha}{D^2},
$
which proves \eqref{eq:Wamp}. In particular, the identity of the amplified sources enters only through total broadcast precision $T_H$.

Differentiating with respect to $T_H$ gives
$
\operatorname{sign}\left(\frac{dW}{dT_H}\right)
=
\operatorname{sign}\big(T_H+\alpha+\beta(1-r)(1-2r)\big).
$
The most negative value of $(1-r)(1-2r)$ is $-1/8$, attained at $r=3/4$. Hence $\alpha\ge\beta/8$ implies $dW/dT_H\ge0$ for every $T_H\ge0$, with equality only at $\alpha=\beta/8$, $T_H=0$, and $r=3/4$. Away from that knife edge, accuracy is strictly increasing in $T_H$. An exact-use budget of $m$ broadcast slots is therefore maximized by choosing the $m$ largest precisions in the candidate set. The case $m=1$ gives the most precise single hub.
\end{proof}

\section{Proofs for Section~\ref{sec:bias}}\label{app:bias}

\begin{proof}[Proof of Proposition~\ref{prop:decomp}]
By \eqref{eq:loadings-sum-one}, the signal loadings in every equilibrium add to one. Substituting $x_k=\theta+\varepsilon_k$ and $y=\theta+\eta$ into \eqref{eq:linear-strategy} therefore gives
$
a_i-\theta
=
\sum_k c_{ik}\varepsilon_k+c_{iy}\eta+d_i.
$
The noise terms are independent and have mean zero. Hence
$
\E[(a_i-\theta)^2]
=
\sum_k\frac{c_{ik}^2}{\beta_k}
+\frac{c_{iy}^2}{\alpha}
+d_i^2.
$
Averaging across users and changing sign gives \eqref{eq:decomp}. Theorem~\ref{thm:exist} shows that the loadings do not depend on $b$, so directional preferences affect accuracy only through $d_i^2$.

Under uniform coordination and common $r$, the fixed point is
$
d_i=(1-r)b_i+\frac{r}{N-1}\sum_{j\ne i}d_j.
$
Summing across $i$ gives $\sum_i d_i=\sum_i b_i$. Substituting this identity back into the fixed point and taking $N\to\infty$ yields
$
d_i\to(1-r)b_i+r\bar b $, $
\bar b=\frac1N\sum_jb_j.
$
\end{proof}

\begin{proof}[Proof of Lemma~\ref{lem:polarized-communities}]
Write $N=2n$. In the integrated network, every user observes the public signal and all $N$ private signals. Hence all users share the full-information posterior mean
$
\widehat\theta^{I}
=
\frac{\alpha y+\beta\sum_{k=1}^{N}x_k}
{\alpha+N\beta},
$
whose posterior variance is $1/(\alpha+N\beta)$. Because coordination is uniform, symmetry implies that users in the $+b$ group have a common directional component $d_+^{I}$ and users in the $-b$ group have $d_-^{I}=-d_+^{I}$. For a user in the $+b$ group,
$
d_+^{I}
=(1-r)b
+\frac{r}{N-1}\left[(n-1)d_+^{I}+n d_-^{I}\right]
=(1-r)b-\frac{r}{N-1}d_+^{I}.
$
Therefore
$
d_+^{I}
=\frac{1-r}{1+r/(N-1)}b
$, $
d_-^{I}=-d_+^{I},
$
which gives
$
W^{I}
=-\frac{1}{\alpha+N\beta}
-\left(\frac{1-r}{1+r/(N-1)}\right)^2b^2.
$

In the segregated network, every user observes the public signal and all $n$ private signals generated in her own group. Users within a group therefore share a posterior mean with variance $1/(\alpha+n\beta)$. Since all coordination weight is placed on same-group users and all users in a group have the same bias, the directional fixed point for a group with sign $s\in\{-1,+1\}$ is
$
d_s^{S}=(1-r)sb+r d_s^{S},
$
so $d_s^{S}=sb$. Hence
$
W^{S}
=-\frac{1}{\alpha+n\beta}-b^2.
$
Finally, $\alpha+N\beta>\alpha+n\beta$ and
$
0<\frac{1-r}{1+r/(N-1)}<1,
$
so both the variance term and the squared directional term are strictly smaller in the integrated network. Therefore $W^{I}>W^{S}$.
\end{proof}

\begin{proof}[Proof of Proposition~\ref{prop:general-score}]
From \eqref{eq:bias-solution}, $d(M)=L(M)b$, where
$
L(M)=(I-R\Phi(M))^{-1}(I-R).
$
Using Proposition~\ref{prop:decomp},
$
W(M)
=-\bar V(M)-\frac1N\|L(M)b\|^2
=-\bar V(M)-b'\left(\frac1N L(M)'L(M)\right)b.
$
Since $\mathcal Q(M)=N^{-1}L(M)'L(M)$, this is \eqref{eq:general-bias-accuracy}. Taking the difference between $M_1$ and $M_0$ gives
$
W(M_1)-W(M_0)
=
\bar V(M_0)-\bar V(M_1)
-b'\big(\mathcal Q(M_1)-\mathcal Q(M_0)\big)b,
$
which is \eqref{eq:general-score}.

To obtain the overlap representation, write $L(M)=[\ell_1(M),\ldots,\ell_N(M)]$. Then
$
b'\mathcal Q(M)b
=
\frac1N\left\|\sum_c b_c\ell_c(M)\right\|^2,
$
and expanding the square gives \eqref{eq:overlap-decomp}. Finally, $R\Phi(M)$ is nonnegative and has spectral radius below one, so
$
(I-R\Phi(M))^{-1}=\sum_{k=0}^\infty(R\Phi(M))^k
$
is entrywise nonnegative. Hence every influence profile $\ell_c(M)$ is nonnegative and every overlap coefficient $\mathcal C_{ch}(M)$ is nonnegative. The sign of each cross term is therefore determined by $b_cb_h$.
\end{proof}

\begin{proof}[Proof of Corollary~\ref{cor:robust}]
By Proposition~\ref{prop:general-score}, an intervention improves accuracy exactly when
$
b'\Delta\mathcal Qb\le\Delta V.
$
If $\Delta\mathcal Q\succeq0$, this is the stated quadratic region. If $\Delta\mathcal Q\succ0$ and $\Delta V>0$, the region is an ellipsoid; if $\Delta V=0$, positive definiteness implies that only $b=0$ satisfies the inequality.

For the norm-bounded case, Rayleigh--Ritz gives
$
\sup_{\|b\|\le B}b'\Delta\mathcal Qb
=
B^2\max\{\lambda_{\max}(\Delta\mathcal Q),0\}.
$
The maximum with zero is needed because the admissible ball contains $b=0$. Hence the intervention raises accuracy for every $\|b\|\le B$ if and only if \eqref{eq:robust-score} holds. If $\Delta\mathcal Q$ is indefinite, its quadratic form takes both positive and negative values along eigenvectors associated with positive and negative eigenvalues, so there is no scalar topology-free threshold without additional restrictions on $b$.
\end{proof}

\begin{proof}[Proof of Corollary~\ref{cor:source-ranking}]
Apply Proposition~\ref{prop:general-score} to the pair $(M,M^{(c)})$. The resulting accuracy change is
$
W(M^{(c)})-W(M)=\mathcal S_c(M;b).
$
Thus the platform ranks mutually exclusive one-source interventions by $\mathcal S_c(M;b)$. If every score is strictly negative, no amplification strictly dominates every candidate; if the maximal score is zero, no intervention and the zero-score candidates are tied. The dependence on the entire bias vector follows from the off-diagonal terms in \eqref{eq:overlap-decomp}.
\end{proof}

\begin{proof}[Proof of Proposition~\ref{prop:multiplier}]
If source $c$ is the only biased user, then $b=b_ce_c$. Equation~\eqref{eq:bias-solution} gives
$
d(M)
=b_c(I-R\Phi(M))^{-1}(I-R)e_c
=b_c\ell_c(M).
$
Therefore
$
\frac1N\|d(M)\|^2
=b_c^2\frac1N\|\ell_c(M)\|^2
=b_c^2\mathcal B_c(M).
$
Substituting into the bias--variance decomposition yields
$
W(M)=-\bar V(M)-b_c^2\mathcal B_c(M).
$
Taking the difference across $M_0$ and $M_1$ gives
$
W(M_1)-W(M_0)
=
\Delta V-b_c^2\Delta\mathcal B_c.
$
All cases in the proposition follow directly. In particular, if $\Delta\mathcal B_c>0$ and $\Delta V>0$, amplification is accuracy improving exactly when
$
|b_c|\le\sqrt{\frac{\Delta V}{\Delta\mathcal B_c}}.
$
\end{proof}

\begin{proof}[Proof of Proposition~\ref{prop:bonacich}]
Because each row of $\Phi(M)$ sums to one,
$
\|R\Phi(M)\|_\infty
=
\max_i r_i\sum_{j\ne i}\phi_{ij}(M)
=
\bar r<1.
$
Hence $\rho(R\Phi(M))<1$ and the Neumann series converges:
$
(I-R\Phi(M))^{-1}
=
\sum_{k=0}^{\infty}(R\Phi(M))^k.
$
Since $(I-R)e_c=(1-r_c)e_c$,
$
\ell_c(M)
=(1-r_c)\sum_{k=0}^{\infty}(R\Phi(M))^ke_c.
$
The $(i,c)$ entry of $(R\Phi(M))^k$ is the total weight of all directed coordination walks of length $k$ from $i$ to $c$, where the weight of a walk is the product of its edge weights and coordination coefficients. Summing over $k$ therefore gives the discounted influence of source $c$ on user $i$. Squaring these components and averaging over users gives $\mathcal B_c(M)$, establishing the Katz--Bonacich interpretation.
\end{proof}

\begin{proof}[Proof of Proposition~\ref{prop:demagogue}]
We evaluate the exact single-source multiplier criterion in the canonical local topology and then take the large-$N$ limit. Followers are unbiased and the hub has bias $b_c$. Let $x$ denote the common follower residual and $d_c$ the hub residual. The bias fixed point gives
$
x=r\left[q\frac{d_c+2x}{3}+(1-q)x\right] $, $
d_c=(1-r)b_c+rx.
$
Substituting the second equation into the first and collecting terms yields
$
x
=
\frac{rq(1-r)}{3A-r^2q}b_c $, $
A=1-r+\frac{rq}{3},
$
which is \eqref{eq:xres}.

Under the sender topology, followers have asymptotically unit population share, so the average bias channel is $x^2+O(1/N)$; the hub's own contribution is only $d_c^2/N$. Under the receiver topology, the hub is not a salient observed source for followers. Its contribution enters only through the vanishing non-neighbor weight, so the average follower bias channel is $O(1/N)$. Hence the change in the average bias channel converges to $x^2$.

Let $\Delta V>0$ denote the variance gain from broadcasting. To leading order, amplification improves accuracy if and only if
$
\Delta V\ge x^2.
$
Substituting the expression for $x$ and solving for $|b_c|$ gives
$
|b_c|
\le
\sqrt{\Delta V}\,\frac{3A-r^2q}{rq(1-r)}.
$
Finally,
$
3A-r^2q
=3(1-r)+rq-r^2q
=(1-r)(3+rq),
$
so
$
b_c^*
=
\sqrt{\Delta V}\,\frac{3+rq}{rq}.
$
This is \eqref{eq:threshold}.
\end{proof}

\subsection{Finite-network validation of the canonical threshold}\label{app:bias-validation}
The closed-form threshold uses the large-$N$ approximation for the bias coefficient, while the variance gain is recomputed at each finite $N$. Holding $(r,q,\alpha,\beta)=(0.85,0.7,2.5,2.5)$ and the local-ring topology fixed, let
\[
\Delta V_N=\bar V_N(M_{\mathrm{receiver}})-\bar V_N(M_{\mathrm{sender}}).
\]
Inserting the exact $\Delta V_N$ into \eqref{eq:threshold} gives $0.419$ at $N=16$ and $0.699$ at $N=100$. The exact finite-$N$ crossovers are $0.408$ and $0.695$, respectively. The change in the threshold across $N$ is therefore driven by the finite-$N$ variance gain; the limiting bias coefficient $rq/(3+rq)$ is held fixed. The exact finite-$N$ bias coefficient exceeds its limiting value by about $5.45\%$ at $N=16$ and $1.21\%$ at $N=100$.

\section{Proof for Section~\ref{sec:engagement}}\label{app:engagement-proof}

\begin{proof}[Proof of Proposition~\ref{prop:divergence}]
Holding the signal-related terms fixed, the accuracy objective is
$
W=-\frac1N\sum_i(V_i+d_i^2),
$
so its bias component is uniquely maximized at $d=0$.
For engagement, \eqref{eq:engagement} gives the bias-dependent loss
$
\mathcal L(d)
=
\sum_i(1-r_i)(d_i-b_i)^2
+
\sum_i r_i\sum_j\phi_{ij}(d_j-d_i)^2.
$
Maximizing $\Pi$ is equivalent to minimizing $\mathcal L$. Differentiating with respect to $d_i$ and dividing by two gives
\begin{equation}\label{eq:engagement-foc-app}
(1-r_i)(d_i-b_i)
+r_i\sum_{j\ne i}\phi_{ij}(d_i-d_j)
+\sum_{h\ne i}r_h\phi_{hi}(d_i-d_h)=0.
\end{equation}
The first term pulls $d_i$ toward the user's own ideal $b_i$. The second term pulls it toward users whom $i$ references, and the third term reflects the utilities of users who reference $i$.

The quadratic loss $\mathcal L$ is strictly convex. The coordination terms are positive semidefinite, while the own-action terms have strictly positive coefficients $1-r_i$. Hence \eqref{eq:engagement-foc-app} has a unique solution. The coefficient matrix has positive diagonal entries, nonpositive off-diagonal entries, and is strictly diagonally dominant by $1-r_i>0$; it is therefore a nonsingular symmetric $M$-matrix with an entrywise nonnegative inverse. In a connected component whose nonzero ideals all have the same sign, the unique engagement-optimal residuals inherit that sign. Unless all relevant ideals are zero, $d=0$ is therefore not engagement optimal. This proves that accuracy and engagement evaluate the bias channel differently.
\end{proof}

\section{Simulation of the liberal planner}\label{app:sim}

For each current network, the platform evaluates one candidate addition and one candidate deletion for every user. Users then choose among the proposed changes according to their own expected utility, and accepted row changes are applied simultaneously. The update index is algorithmic time: users are myopic, there is no intertemporal state variable other than the network, and no signal is carried across rounds.

\begin{algorithm}[ht]
\caption{Liberal-planner recommendation dynamics (one update round)}\label{alg:planner}
\KwIn{network $M$; parameters $(\alpha,\{\beta_i\},r,q,\{b_i\})$}
Compute baseline accuracy $W_0$ from $M$\;
\For{each user $i$}{
  \emph{addition:} among missing links $i\to j$ whose addition raises accuracy above $W_0$, choose the $j$ that produces the largest $W$; if no addition raises $W$, propose none\;
  \emph{deletion:} among links $i\to j$ whose deletion raises accuracy above $W_0$, choose the $j$ that produces the largest $W$; if no deletion raises $W$, propose none\;
}
\For{each user $i$}{
  evaluate the available accept/reject combinations for the proposed addition and deletion using user $i$'s expected utility at the current network $M$, and record the utility-maximizing row for user $i$\;
}
Apply all recorded row updates simultaneously\;
\Return updated network $M$\;
\end{algorithm}

Each candidate change is screened in isolation against the current network. Because accepted changes are implemented simultaneously, the platform objective need not improve monotonically from one update round to the next. A rest point is instead defined by the absence of any recommendation that would be proposed and accepted at the prevailing network. Friends-of-friends chooses additions with maximal overlap in observed sources and deletions with minimal overlap; the random benchmark samples eligible links uniformly.

The baseline calibration is $r=0.7$, $q=0.5$, $\alpha=30$, $\beta_i\sim\mathcal U\{10,\ldots,45\}$, $b_i\in\{-0.6,+0.6\}$, initial link probability $0.4$, and 60 update rounds. Figure~\ref{fig:formation} uses 100 paired runs at each of $N=20$ and $N=50$. The full candidate set is evaluated at $N=20$; at $N=50$ the platform screens six randomly sampled candidate additions and deletions per user. Figure~\ref{fig:engagement} uses the same six-candidate screen at both sizes. Pairing holds the initial network, precision draw, and random seed fixed across recommendation rules.

The equilibrium solver iterates the contraction implied by the best-reply system and solves the precision-independent intercept fixed point exactly. Baseline equilibria use tolerance $10^{-9}$ and candidate networks use warm starts. The replication package contains the code, 100-run series, confidence intervals, end-state decompositions, and two-camp parameter grids. In a separate scan of accuracy-improving additions, unbiased users accept $1812/1871=97\%$ of proposals. With $b_i\in\{-0.6,+0.6\}$, users accept $427/1320=32\%$; $659$ of the $893$ rejected proposals are opposite-sign links.

\subsection{Additional robustness for implementation and engagement}\label{app:engagement-robust}

The main text reports the central implementation statistics. In a separate scan of accuracy-improving additions, unbiased users accept $1812/1871=97\%$ of proposals. With $b_i\in\{-0.6,+0.6\}$, users accept $427/1320=32\%$; $659$ of the $893$ rejected proposals are opposite-sign links. Extending the recommendation horizon to 200 update rounds and increasing network size from $N=20$ to $N=50$ preserve the qualitative ordering of the recommendation rules.

For the engagement comparison, we also solve a finite-$N$ two-camp exercise under three $(r,q)$ calibrations: the baseline $(0.7,0.5)$, a strong-coordination variant $(0.85,0.7)$, and a weak-coordination variant $(0.5,0.3)$. There are $N=20$ users, ten in each bias camp. The intervention broadcasts a positive-bias candidate hub to the remaining positive-bias users while holding pre-existing links fixed. We set $\alpha=30$, ordinary precision $\beta=20$, and ordinary ideals $b_i=\pm0.6$, and vary the hub's precision and bias. Across all three calibrations there is a nonempty but bounded region with $\Delta W<0<\Delta\Pi$. For a precise hub ($\beta_c=45$) at the baseline, both objectives rise for $b_c\le0.50$, the disagreement region is $b_c\in[0.53,0.95]$, and beyond that interval $\Delta\Pi$ becomes negative as validation gives way to overshooting and coordination conflict. Strong coordination narrows the interval to $[0.60,0.70]$, while weak coordination widens it to $[0.58,1.23]$. Full sign maps and component decompositions are included in the replication materials.

To verify the decoupling mechanism numerically, we hold coordination weights uniform while changing the information network. The bias profile is then invariant to informational rewiring, as implied by $d=(I-R\Phi)^{-1}(I-R)b$. In the corresponding recommendation simulations, the engagement-minus-planner polarization gap is zero up to numerical precision. These exercises are robustness checks for the analytical channel decomposition rather than separate equilibrium results.

\end{document}